\documentclass[11pt]{article}
\usepackage{glauber}
\usepackage{physics}
\usepackage{mathtools}

\usepackage{rgd}
\usetikzlibrary{arrows.meta,positioning}
\newlength{\capht}
\allowdisplaybreaks

\begin{document}

\title{Markov Chains Approximate Message Passing}
\author{Amit Rajaraman\thanks{MIT. \texttt{amit\_r@mit.edu}. Supported by a MathWorks Fellowship.} \and
    David X. Wu\thanks{UC Berkeley. \texttt{david\_wu@berkeley.edu}. Supported by NSF Graduate Research Fellowship DGE-2146752.}}
\date{\today}
\maketitle
\begin{abstract} 
    Markov chain Monte Carlo algorithms have long been observed to obtain near-optimal performance in various Bayesian inference settings.
    However, developing a supporting theory that makes these studies rigorous has proved challenging.
    
    In this paper, we study the classical spiked Wigner inference problem, where one aims to recover a planted Boolean spike from a noisy matrix measurement. 
    We relate the recovery performance of Glauber dynamics on the annealed posterior to the performance of Approximate Message Passing (AMP), which is known to achieve Bayes-optimal performance.
    Our main results rely on the analysis of an auxiliary Markov chain called \emph{restricted Gaussian dynamics} (RGD). 
    Concretely, we establish the following results:
    \begin{enumerate}
        \item RGD can be reduced to an effective one-dimensional recursion which mirrors the evolution of the AMP iterates.
        \item From a warm start, RGD rapidly converges to a fixed point in correlation space, which recovers Bayes-optimal performance when run on the posterior.
        \item Conditioned on widely believed mixing results for the SK model, we recover the phase transition for non-trivial inference.
    \end{enumerate}
\end{abstract}

\thispagestyle{empty}
\setcounter{page}{0}
\newpage
\tableofcontents
\thispagestyle{empty}
\setcounter{page}{0}
\newpage


\section{Introduction}
Markov chains are fundamental objects of study which model the dynamics of complex systems arising in statistical physics, computational biology, and social networks.
Besides modeling physical systems, they are a highly successful algorithmic tool to sample from high-dimensional distributions in Bayesian inference, statistics, and theoretical computer science.

It has long been empirically observed that canonical Markov chains such as \emph{Glauber dynamics} which attempt to sample from the posterior distribution are quite successful for Bayesian inference. For example, \cite{DKMZ11} observed that for community detection in the stochastic block model, the output of the Glauber dynamics is Bayes optimal.
However, there has been a dearth of rigorous results confirming or explaining the algorithmic success of Markov chains in statistical inference. 
One fundamental reason for the lack of rigorous progress is that, for many posterior distributions of interest, Glauber dynamics provably suffers from slow worst-case mixing. This leads to the following natural question.

\begin{question}\label{question:glauber-opt}
    If we run Glauber dynamics for polynomial time on the posterior distribution of a natural Bayesian inference problem, does it achieve Bayes-optimal performance?
\end{question}

To make progress on this, we revisit early empirical studies which compared the performance of Glauber dynamics to that of local \emph{message passing} algorithms.
Our starting point is \cite{DKMZ11}, which found that the performance of Glauber dynamics matches that of the Belief Propagation (BP) algorithm which directly estimates the marginals of the posterior in the stochastic block model.

Since then, message passing algorithms such as BP and its dense variant Approximate Message Passing (AMP) \cite{DMM09,BM11} have been widely applied in a variety of high-dimensional statistics problems, to great success (for a few examples, see \cite{DAM16,MV21,EAMS22,HMP24} and references therein). 
A particularly appealing property of AMP is that its asymptotic performance can be rigorously characterized by a finite-dimensional recursion known as \emph{state evolution}.

For certain inference problems, one can use state evolution to design AMP algorithms to construct Bayes-optimal estimators. 
Similarly, we \emph{a priori} know that Glauber achieves the same guarantees in exponential time, since by then the dynamics would be able to produce Gibbs samples.
Motivated by this, we refine our initial question to ask about connecting the performance of these two classes of algorithms if Glauber is only run for polynomial time.
\begin{question}\label{question:glauber-amp}
    If we run Glauber dynamics for polynomial time, can we prove it performs as well as AMP?
\end{question}

Our main results make progress on \Cref{question:glauber-amp} by proving a formal connection between Markov chains and AMP for the spiked Wigner inference problem, which we define below.
We show that the dynamics of the correlation with the spike---the observable of primary interest---are governed by fixed point relations which also characterize the asymptotic performance of AMP for the same problem. See \Cref{fig:main-results} for a schematic diagram of our main results.

\begin{figure}[!ht]
    \centering
    
\begin{tikzpicture}[
  box/.style={draw, rounded corners, minimum width=3.6cm, minimum height=1.2cm, align=center},
  >={Stealth[length=3mm]},
  node distance=1.8cm and 4.8cm
]

\node[box] (glauber) {Glauber};
\node[box, below=of glauber] (rgd) {RGD};
\node[box, right=of glauber] (amp) {AMP/BP};
\node[box, below=of amp] (se) {State Evolution};

\draw[<->, thick] (glauber) -- node[midway, left, align=center]{Partial answer \\ in \cite{LMRRW24}} (rgd);
\draw[<->, thick] (amp) -- node[midway, right] {e.g. \cite{BM11}} (se);

\draw[dashed, <->, thick] (glauber) -- node[midway, above, align=center] {Predicted by \cite{DKMZ11}} (amp); 
\draw[<->, thick] (rgd) -- node[midway, above] {\textbf{This work}} (se);               

\end{tikzpicture}
    \caption{An illustration of our results.}
    \label{fig:main-results}
\end{figure}
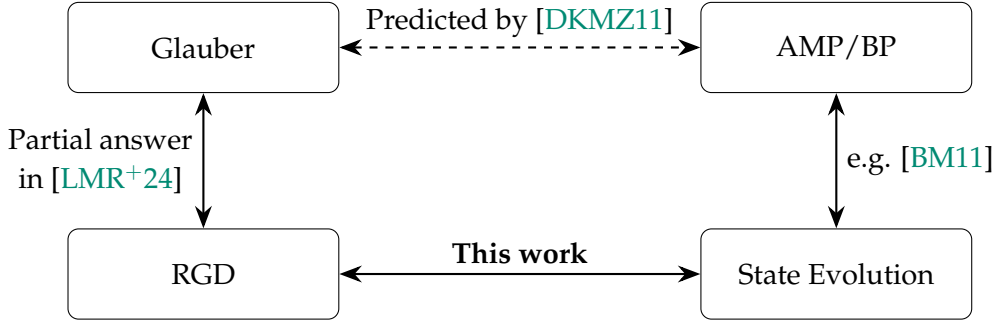

The precision of our results requires model-specific calculations that we did not attempt to generalize beyond the spiked Wigner setting.
However, our high-level proof strategy is fairly generic and we believe it should be applicable to general spiked matrix inference problems.

In the spiked Wigner model, one aims to recover an unknown \emph{signal} $\bx$ drawn from the prior $\mathrm{Unif}(\{\pm 1\})^{\otimes N}$ given a noisy matrix measurement 
\begin{equation}
   \label{eq:def-spiked-wigner}
   \bM = \bW + \frac{\lambda}{N} \bx \bx^\top \mcom
\end{equation}
for some signal-to-noise ratio (or SNR) $\lambda \ge 0$, and a \emph{noise} matrix $\bW \sim \GOE(N)$ with $\bW_{ij} = \bW_{ji} \sim \calN\left( 0 , \frac{1}{N} \right)$ for $i \neq j$ and $\bW_{ii} \sim \calN(0, \tfrac{2}{N})$.

The spiked Wigner model has been extensively studied over the last two decades, providing a rich test bed for average-case complexity. Perhaps the first interesting phenomenon uncovered by this line of work is the famous BBP phase transition \cite{HR04,BBAP05}, which marks the information-theoretic and computational phase transition at $\lambda = 1$. Indeed, it is known that for $\lambda < 1$, $\bM$ is mutually contiguous with $\bW$, whereas for $\lambda > 1$ the top eigenvector of $\bM$ is nontrivially correlated with $\bx$.

A straightforward calculation shows that the posterior distribution in the spiked Wigner model is of the form
\begin{equation}
   \label{eq:posterior-spiked-wigner}
   \mu_{\beta \bM}(\sigma) \propto \exp\left( \frac{\beta}{2} \sigma^\top \bM \sigma \right) = \exp\left( \frac{\beta}{2} \sigma^\top \bW \sigma + \frac{\beta\lambda}{2N} \langle \sigma , \bx\rangle^2 \right)
\end{equation}
for $\beta = \lambda$. 

\paragraph{Beyond worst-case mixing.}
\cite{LMRRW24} introduced the framework of locally stationary distributions to analyze the performance of Markov chains before they have mixed. 
Roughly speaking, a locally stationary distribution is any distribution which is approximately invariant under another step of the Markov chain. 
They are also very easy to sample from---any
reversible Markov chain such as Glauber dynamics produces a locally stationary distribution in $\poly(N)$ steps.

For the spiked Wigner problem, they show that for sufficiently large (constant) $\lambda$ and $\beta < \frac{1}{4}$, the Glauber dynamics run on $\mu_{\beta\bM}$ for a large polynomial time from an arbitrary initialization reaches a vector that is nontrivially correlated with the spike. 
To prove this result, they combined the properties of locally stationary distributions with an analysis of a different Markov chain known as \emph{Restricted Gaussian Dynamics} (RGD), which we define below.\footnote{The precise definition is not terribly important, so readers seeking intuition can skip the definition without much concern.}

\begin{definition}[{Restricted Gaussian dynamics \cite{STL20,LST21,CE22}}]
   \label{def:rgd}
    For $\bM = \bW + \frac{\lambda}{N} \bx\bx^{\top}$, \emph{restricted Gaussian dynamics} ($\RGD$) at inverse temperature $\beta$ is a Markov chain on $\{\pm1\}^N$ where a transition from $\sigma$ to $\bsigma'$ is given by the following:
    \begin{itemize}
        \item Sample $\bg\sim\calN(0,1)$ and define $\bz \coloneqq \beta\lambda \frac{1}{N}\langle \sigma,\bx\rangle + \sqrt{\frac{\beta\lambda}{N}} \bg$.
        \item Sample $\bsigma'$ from the distribution $\mu_{\beta \bW, \bz \bx}$.
    \end{itemize}
\end{definition}
It is straightforward to show that RGD has $\mu_{\beta \bM}$ as its unique stationary distribution. 
However, the same bottleneck that precludes fast mixing for Glauber also prevents RGD from mixing quickly.
In spite of this, we achieve a precise and nearly distributional understanding of the dynamics of RGD, which allows us to deduce several interesting properties of its dynamics. 
Although RGD is not an implementable algorithm, the tools developed by \cite{LMRRW24} allow us to transfer this understanding over to Glauber dynamics in certain temperature regimes.

Before we state our main results, we take a brief excursion to give some standard spin glass definitions.
\begin{definition}
   Let $\bW \sim \GOE(N)$. The Sherrington--Kirkpatrick model (SK model) at inverse temperature $\beta > 0$ and external field $v \in \R^N$ is the distribution over $\{\pm 1\}^N$ with density
   \[ \mu_{\beta \bW, v}(\sigma) \propto \exp\left( \frac{\beta}{2} \cdot \sigma^\top \bW \sigma + \langle v,\sigma\rangle \right) \mper \]
\end{definition}
Our results for the performance of Glauber on the spiked Wigner posterior rely on the rapid mixing of Glauber dynamics for the SK model.  
To simplify the discussion, we will state our results whenever the SK model satisfies the following mixing condition based on modified log-Sobolev inequalities.

\begin{condition}[\textnormal{\textsf{MLSI}}]
   \label{cond:sk-beta-one}
   We say that $\beta < 1$ satisfies \condref{cond:sk-beta-one} if with probability $1-o(1)$ over the random matrix $\bW$, for all $h \in \R$, $\mu_{\beta \bW , h \bone}$ satisfies a modified log-Sobolev inequality with constant $\Omega(1/N)$.
\end{condition}

A recent breakthrough line of work \cite{EKZ22,CE22,AKV24} has made significant progress towards establishing rapid mixing of Glauber dynamics in the entire high-temperature regime for the SK model. In particular, they establish that \condref{cond:sk-beta-one} holds for $\beta < 0.295$; it is conjectured that \condref{cond:sk-beta-one} holds for the entire high-temperature regime $\beta < 1$ \cite{MPV87}. We are now poised to state our main results. 

\subsection{Main results} 

The following theorem establishes that the short-term behavior of the correlation of Glauber iterates is an approximate fixed point of an AMP-like fixed point equation in the entire high SNR regime $\lambda > 1$. 
Furthermore, it establishes a sharp phase transition at $\beta = \frac{1}{\lambda}$ for the success of Glauber dynamics at weak recovery of the spike.

\begin{theorem}[Informal, see \Cref{thm:glauber-formal}]
   \label{thm:glauber-spiked-wig}
   Let $\bsigma_0 \in \{ \pm 1 \}^N$ be arbitrary, and let $(\bsigma_t)_{t \ge 0}$ be the trajectory of the Glauber dynamics run with stationary distribution $\mu_{\beta\bM}$ initialized at $\bsigma_0$. Let $T \ge \wt{\Omega}(N^{4})$, and $\wh{\bsigma} = \bsigma_{\bt}$ for a uniformly random $\bt \sim [0,T]$. 
   
    Then, for any $\eps > 0$ and any $\beta \in (0, 1)$ such that $\beta \neq \frac{1}{\lambda}$ and \condref{cond:sk-beta-one} holds, with high probability over the noise $\bW$, we have
   \[ \Pr \left[ \left| \frac{1}{N} \left| \langle \wh{\bsigma} , \bx \rangle  \right| - \OPT_{\beta,\lambda} \right| > \frac{1}{N^{1/2 - \eps}} \right] = o_N(1) \mcom \]
   for some constant $\OPT_{\beta,\lambda}$ that is the largest fixed point of an explicit AMP-like recursion. Here, $\OPT_{\beta,\lambda} > 0$ if $\beta > \frac{1}{\lambda}$ and $\OPT_{\beta,\lambda} = 0$ if $\beta < \frac{1}{\lambda}$. 
   
   Furthermore, $\OPT_{\beta, \lambda}$ characterizes the correlation of a Gibbs sample:   
   \[ \Pr_{\bsigma \sim \mu_{\beta\bM}} \left[ \left| \frac{1}{N} | \langle \bsigma,\bx\rangle | - \OPT_{\beta,\lambda} \right| > \frac{1}{N^{1/2 - \eps}} \right] \le e^{-cN^{\eps}} \mper \]
\end{theorem}

\begin{remark}
    \label{rem:lsd-speedup-n3}
    If one only cared about obtaining a correlation that is within $o_N(1)$ of $\OPT_{\beta,\lambda}$, the formal version of the above result will show that it suffices to set $T = \wt{\Omega}(N^3)$.
\end{remark}

The key technical innovation in our work is a much more fine-grained understanding of the RGD Markov chain, which is essential for proving the above theorem. As mentioned earlier, the connection between RGD and the Glauber dynamics is present in prior work \cite{LMRRW24}. Although the form of their connection breaks when $\beta > 1$, our understanding of RGD extends to this regime. We believe that a connection between RGD and the Glauber dynamics should also exist in this low-temperature regime, but do not attempt to answer this question in the present work.

Towards illustrating our fine-grained understanding of RGD, we present the following theorem, which provides a direct comparison between RGD and AMP. Indeed, this connection is precisely what yields the fixed point equations alluded to in \Cref{thm:glauber-spiked-wig}.

\begin{theorem}[Informal, see \Cref{subsec:rgd-amp-simulation}]
   \label{thm:informal-rgd-amp}
    Let $\lambda \ge \beta > 1$ and $\eps > 0$ independent of $N$. There exists some threshold $\tau \in (0,1)$ (depending on $\beta,\lambda$) such that the following is true. 
    
    Let $x \in \{\pm 1\}^N$ be an arbitrary spike, and fix $\sigma \in \{\pm 1\}^N$ such that $\frac{1}{N} \left|\langle \sigma , x\rangle\right| > \tau$. Let $\bsigma^{\RGD}$  be obtained by taking one step of RGD from $\sigma$, and $\bsigma^{\AMP}$ be obtained by taking one step of an AMP-like update from $\sigma$.
    Then, with high probability over the noise $\bW$, the following holds:
   \[ \left| \frac{1}{N} \langle \bsigma^{\RGD} , x \rangle - \frac{1}{N} \langle \bsigma^{\AMP} , x \rangle \right| = o_N(1) \mper \]
\end{theorem}
\begin{remark}
    The formal version of this theorem depends on verifying an explicit numerical condition which can easily be simulated on a computer; see \Cref{subsec:rgd-amp-simulation} for more details.
\end{remark}
In other words, if both algorithms are provided a sufficiently correlated warm start, one step of RGD is close to one step of AMP in terms of correlation. In the formal version, the instantiation $\sigma$ is also allowed to slightly depend on $\bW$.

One is naturally led to wonder whether a version of \Cref{thm:glauber-spiked-wig} holds true for RGD  for $\beta > 1$, so that we can relate the correlation of a locally stationary point of RGD to that of a sample from the posterior? Unfortunately, we expect such results to be false from worst-case initializations. However, we \emph{can} still hope for such results from non-worst-case initializations, specifically from warm starts that already have some non-trivial correlation.

\begin{theorem}[Informal, see \Cref{th:rgd-fixed-temp}]
   Fix $\lambda > 1$, $K > 0$ large, $\eps > 0$ small, and $\frac{1}{\lambda} < \beta \le \lambda$. There exists a threshold $\tau \in (0,1)$ such that the following holds. With high probability over the noise $\bW$, the following holds. Let $\sigma_0 \in \{\pm 1\}^N$, and assume that $|R(\sigma_0,\bone)| > \tau$. Suppose that we run the RGD Markov chain at inverse temperature $\beta$ from $\sigma_0$ for $T \ge \omega(\log N)$ steps to arrive at distribution $\nu_T$. Suppose $T \le N^{K}$. Then,
	\[ \E_{\bsigma_T \sim \nu_T}\left[ \left| |R\left( \bsigma_T , \bx \right)| - \OPT_{\beta,\lambda} \right| \right] \le O\qty(\frac{1}{N^{1/2 - \eps}}) \mper \]
\end{theorem}

While the requirement that the initialization has sufficiently large correlation may seem like too strong a requirement, we emphasize that this is not the case. If one ran an annealed version of the RGD Markov chain where $\beta$ is slowly incremented, the resulting trajectory of correlations is large enough that this lower bound condition is satisfied. Accordingly, we believe this might provide a path to showing that annealed \emph{Glauber} dynamics attains a correlation of $\OPT_{\beta,\lambda}$.

\begin{remark}
    Again, the formal version of the above theorem requires establishing two explicit numerical conditions, which can easily be simulated on a computer.
\end{remark}

\begin{remark}
We draw attention here to the (perhaps) surprising fact that simulations suggest the correlation $\OPT_{\beta,\lambda}$ with the spike increases \emph{past} $\beta = \lambda$.
This is a rather subtle point: while the posterior mean, corresponding to $\mu_{\lambda\bM}$ (and AMP for that matter) attains the MMSE, this does not imply that samples from the posterior are Bayes-optimal with respect to correlation. 
What \emph{is} true is that the posterior mean is Bayes optimal for the \emph{normalized squared correlation}, i.e. it is the estimator $\wh{x}$ which maximizes $\E\qty[\frac{\ev{\wh{\bx}, \bx}^2}{\norm{\wh{\bx}}^2}]$.
\end{remark}

En route to proving the above result, we prove a new high-precision estimate for the mean magnetization for the SK model under weak external field, which may be of independent interest.
These estimates enable us to  overcome technical difficulties with worst-case analysis of the RGD dynamics for $\beta < 1$.

\begin{theorem}[Informal, see \cref{lem:all-field-mag-conc}]
    Let $\beta < 1$ and $\eps > 0$ sufficiently small. With probability $1-o(1)$ over $W \sim \GOE(N)$, for all $h = O(N^{-(1/4 + \eps)})$, 
        \[ \frac{1}{N} \E_{\bx \sim \mu_{\beta W,h\bone}} \langle\bx,\bone\rangle = \E_{\bg \sim \calN(0,1)}\left[ \tanh\left( \beta\bg\sqrt{q} + h \right) \right] \left( 1 + O\left(\frac{1}{N^{\eps}}\right) \right) + O\left( \frac{1}{N^{1/2 + \eps}} \right) \mcom \]
    where $q$ is the unique solution to $q = \E[\tanh^2(\beta \bg \sqrt{q} + h)]$.
\end{theorem}
In particular, this result establishes the correct asymptotics for the mean magnetization up to the phase transition in weak external field established by \cite{DW23}.


\subsection{Related work}\label{sec:related-work}

\paragraph{Spiked Wigner model.}

The spiked Wigner model has provided a rich test bed for average-case algorithms and complexity. 
Indeed, the basic question of efficient detection and recovery of the spike has been studied through the lens of spectral algorithms \cite{BBAP05}, low-degree polynomials \cite{PWBM18}, and message passing \cite{DAM16}. Interestingly, these efficient algorithms all succeed in polynomial time once $\lambda > 1$. Furthermore, when $\lambda < 1$, the mutual information between the observation $\bM$ and the signal $\bx$ turns out to be $o(1)$ \cite{LM17}, indicating that one cannot hope to design an estimator $\wh{x} \in \{\pm 1\}^N$ that is non-trivially correlated with the spike $\bx$.

For more general priors beyond $\mathrm{Unif}(\{\pm1\}^N)$, the story becomes significantly richer. 
For example, there exist sparse priors where there exists a statistical-computational gap: there is an intermediate regime of SNRs $(\lambda_c, \lambda)$ where recovery is information theoretically possible, but no known efficient algorithm succeeds at weak recovery. 
\cite{EKJ20} pin down the information-theoretic threshold at $\lambda = \lambda_c$ for general bounded i.i.d. priors. 
The algorithmic side of special cases such as sparse PCA \cite{DM14} and non-negative PCA \cite{MR15} have also been extensively studied. The problem has also been generalized to the tensor setting \cite{MR14,WEM19,BGJ20,KX25}, where there is an interesting gap in performance between naive local algorithms and methods based on the sum-of-squares hierarchy or the Kikuchi method. 

In terms of sampling from the posterior, due to the sign symmetry of the posterior creating a bottleneck in distributions, Glauber suffers an exponential worst-case mixing time.
On the flip side, it has been observed that random initialization or simulated annealing appears to empirically succeed at the goal of sampling from the posterior \cite{DKMZ11}. Nevertheless, it has remained challenging to give provable guarantees about the performance of sampling algorithms, even under algorithmic warm starts. 
A notable positive sampling result in this direction is \cite{MW23}, which uses a different diffusion-based algorithm to sample from the posterior in Wasserstein distance when the SNR $\lambda$ is a sufficiently large constant.

\paragraph{Bayes optimal estimation.}

There is by now a large body of work characterizing Bayes optimal estimation for mean-squared error. Indeed, past work has designed algorithms that attain mean-squared error within $\frac{\polylog(N)}{\sqrt{N}}$ of the MMSE \cite{RV18,LW22}, as well as generalizations to constant-rank spikes, and other priors \cite{Mio17,MV21}. These aforementioned algorithms are based on AMP---one first computes the top eigenvector of $\bM$ to obtain a vector that has non-trivial correlation with the spike $\bx$, then refines this estimate by iteratively applying the AMP update. To our knowledge, these algorithms based on AMP are the only algorithms that are known to obtain a mean-squared error within $o(1)$ of the MMSE in the entire high SNR regime $\lambda > 1$.

Typically, one considers the inference setting where knowledge of the full generative model is assumed. Perhaps confusingly, this is also referred to as the ``Bayes-optimal setting''. 
In \cite{AKUZ19}, they study a mismatched setting for spiked Wigner where the knowledge of $\lambda$ is not assumed. 
In fact, in \Cref{subsec:rgd-amp-simulation} we will show that the performance of their AMP algorithm corresponds with the fixed point equations for RGD.

In certain scenarios, Markov chain-based algorithms (or a sample from the posterior) might be preferred to AMP, since they have the ability to capture second-order information, like the variance of the correlation with the spike. Furthermore, samples can be used more flexibly to compute general posterior expectations. 

\paragraph{Dynamics of Markov chains before mixing.}
Many physically and algorithmically interesting phenomena manifest only when Markov chain dynamics do \emph{not} mix quickly, such as \emph{aging} \cite{CK93} and \emph{metastability} \cite{BH16,BJ24}. 
In the context of \emph{spin glasses}, physicists have made many predictions about the nonequilibrium behavior of natural Markov chains such as Langevin dynamics and Glauber dynamics.
For the pure spherical $p$-spin glass, the seminal work of Cugliandolo and Kurchan \cite{CK93} derived a heuristic set of integro-differential equations which govern the behavior of Langevin dynamics when run for $O(1)$ time. This was later rigorously proven for $T = O(1)$ by \cite{BDG06} and the dynamics of the energy observable was eventually established for $T = \exp(O(N))$ by \cite{Sel24}.
In recent work, \cite{DGPZ25} rigorously prove that the energy and overlap of the sequential scan block dynamics for the SK model are dictated by an analogous set of integro-difference equations for $O(N)$ timescales.

\cite{LSS22} characterized the overlap of Langevin dynamics for spiked matrix posteriors in the large system limit $N \to \infty$ for \emph{fixed} times $t$. 
The analogous result in discrete time would correspond to asymptotically controlling the overlap for $O(N)$ steps of Glauber dynamics. 
In contrast, we give control on any sufficiently large polynomial time scale for Glauber.
\cite{BGP24} give an extremely precise characterization of trajectory of Langevin diffusion for the posterior of the multi-spike tensor PCA problem given polynomially many samples of the spiked tensor (the spike being fixed, and the noise tensor i.i.d. across all samples). 
Concretely, in the rank-one matrix setting, they establish strong recovery of the spike $\bx$ when the SNR $\lambda = N^{\alpha}$ for any $\alpha > 0$.
Their proofs are based on careful analysis of the SDE governing the correlations, which are inaccessible in the discrete setting.

\paragraph{Locally stationary distributions.} Before a recent work of Liu, Mohanty, Raghavendra, Rajaraman, and Wu \cite{LMRRW24}, it was not rigorously known whether Glauber dynamics on the spiked Wigner posterior achieves nontrivial recovery of the spike. As mentioned earlier in the introduction, they introduce the framework of locally stationary distributions to analyze the performance of slow-mixing Markov chains on inference problems such as spiked Wigner. 

Our results directly improve on theirs in the following concrete ways:
\begin{enumerate}[label=(\roman*)]
    \item We achieve tight temperature and SNR thresholds for our analysis, conditional on widely believed mixing conditions for the Sherrington--Kirkpatrick (SK) model. 
    \item Wherever SK mixing holds, we nail down the exact correlation that Glauber dynamics achieves for any SNR $\lambda$, as the solution to a simple fixed point equation. In particular, widely believed mixing conditions for the SK model recover the BBP transition $\lambda > 1$.
    \item In particular, we unconditionally recover weak recovery in the regime $\lambda > \frac{1}{0.295}$ for the success of Glauber, based on the rigorous SK mixing results of \cite{AKV24}.
    \item As in \Cref{rem:lsd-speedup-n3}, we improve the runtime guarantees by an $N^2$ factor.
\end{enumerate}
Locally stationary distributions have since been used in other settings to understand the behavior of Markov chains before mixing. Notably, \cite{BCV25} introduces a definition of local stationarity for quantum Hamiltonians, and uses it to show that all metastable states satisfy certain structural guarantees. \cite{RSSLMRF25} also find an interesting application of locally stationary distributions in the context of stochastic backtracking algorithms for LLM reasoning. 

\subsection{Organization} In \Cref{sec:tech-overview}, we provide a technical overview of our techniques to understand the behavior of RGD, and its connection to AMP. In \Cref{sec:open}, we list a few open questions that seem important to fleshing out our understanding of RGD and Glauber. In \Cref{sec:prelims}, we cover some basic preliminaries that will be useful.

In \Cref{sec:rgd-dynamics}, we prove our main theorems about restricted Gaussian dynamics. Next, in \Cref{subsec:mean-magnetization}, we establish novel high-precision estimates for the mean magnetization of a high-temperature SK model under weak external field, which are crucial for analyzing RGD.
Finally, in \Cref{sec:understanding-fixed-points}, we describe some results analyzing the one-dimensional recursion governing the behavior of RGD.


\section{Technical Overview}\label{sec:tech-overview}

Our main results crucially depend on a very precise understanding of the restricted Gaussian dynamics. 
As it turns out, its study transparently clarifies the fixed point structure and its relationship to AMP. 
Using the machinery of locally stationary distributions \cite[Lemma 3.7]{LMRRW24}, we can then transfer our understanding of RGD to Glauber dynamics; see \Cref{thm:glauber-formal} for more details.

The technical overview is organized as follows.
We will first define RGD and describe the high-level reduction showing why one should expect RGD to boil down to a one-dimensional recursion on correlation. 
Then, in \Cref{subsec:amp-connection}, we show how spin glass machinery can be leveraged to explicitly understand what this one-dimensional recursion is. 
Once we have written down the explicit recursion, we can readily connect it to AMP; this is carried out in \Cref{subsec:rgd-amp-simulation}.
Finally, we describe some qualitative properties of this recursion, such as the nature of its fixed points, in \Cref{subsec:qual}.

First off, observe that due to the rotational symmetry of the noise matrix and prior, we may assume without loss of generality that the spike is the all-ones vector $\bone$. In this case, the scaled posterior \eqref{eq:posterior-spiked-wigner} of the spiked Wigner model becomes
\[ \mu_{\beta \bM}(\sigma) \propto \exp\left( \frac{\beta}{2} \sigma^\top \bM \sigma \right) = \exp\left( \frac{\beta}{2} \sigma^\top \bW \sigma + \frac{\beta\lambda}{2N} \langle \sigma , \bone\rangle^2 \right) \mper \]
We start by explaining where RGD arises from, and the connection between the posterior and the SK model. It arises from the Hubbard--Stratonovich transform \cite{Str57,Hub59} (also see \cite[Theorem 3.12]{LMRW24}).
Let $\bsigma \sim \mu_{\beta \bW + \frac{\beta\lambda}{N} \bone\bone^\top}$, and $\bg \sim \calN\left(0,1\right)$. Set $\bz = \frac{\beta\lambda}{N} \langle \bsigma , \bone\rangle + \sqrt{\frac{\beta\lambda}{N}} \bg$. Then, the posterior distribution on $\bsigma$ conditioned on $\bz$ is simply
\begin{align*}
  \Pr\left[ \sigma \mid \bz \right] &\propto \mu_{\beta\bM}(\sigma) \cdot \exp\left( - \frac{N}{2\beta\lambda} \left( \bz - \frac{\beta\lambda}{N} \langle \sigma,\bone\rangle \right)^2 \right) \\
    &\propto \exp\left( \frac{\beta}{2} \sigma^\top \bW \sigma + \bz \langle \sigma,\bone\rangle \right) = \mu_{ \beta\bW , \bz\bone } \mcom
\end{align*}
and the quadratic term from the spike in the Hamiltonian is now an external field.
In other words,
\[ \mu_{\beta \bW + \frac{\beta\lambda}{N} \bone\bone^\top} = \E_{\bz} \mu_{\beta \bW, \bz\bone} \mper \]
The RGD Markov chain (\Cref{def:rgd}) is then simply the noising-denoising process associated to the above measure decomposition: if at $\sigma \in \{\pm 1\}^N$, draw $\bg \sim \calN\left(0,1\right)$, then move to $\bsigma' \sim \mu_{\beta\bW , \left( \frac{\beta\lambda}{N} \langle \sigma,\bone\rangle + \sqrt{\frac{\beta\lambda}{N}} \bg \right) \bone}$. It is not difficult to show, given this perspective of a noising-denoising process, that RGD has stationary distribution equal to $\mu_{\beta\bM}$. 
The usefulness of RGD is that it brings forth the importance of the correlation with the spike, which is the observable we ultimately care about. 

\paragraph{Reduction to one dimension.}
So far, we have written down the $N$-dimensional RGD chain, whose evolution explicitly depends on the correlation with the spike. 
To make RGD more tractable and eventually recover the AMP update, we would like to reduce the dynamics to a finite-dimensional one. 

This motivates us to consider the related $\PRGD$ Markov chain on $[-1,1]$, which tracks the (normalized) correlation of the $N$-dimensional RGD iterate with the spike. 
In this Markov chain, if starting at $y \in \R$, we  draw $\bg \sim \calN\left(0,1\right)$, set $\bz = \beta\lambda y + \sqrt{\frac{\beta\lambda}{N}} \bg$, draw $\bsigma \sim \mu_{\beta \bW , \bz \bone}$, and finally move to $\bz' = \frac{1}{N} \langle \bsigma,\bone\rangle$. 
The stationary distribution of this Markov chain is just the (normalized) marginal of $\mu_{\beta\bM}$ along the all-ones direction.

So far, this $1$-dimensional chain still goes through an $N$-dimensional SK Gibbs sample $\bsigma$; in particular, it depends on the $N\times N$ matrix $\bW$.
Therefore, to reduce $\PRGD$ to an effective 1-dimensional recursion, we must establish two things:
\begin{itemize}
    \item Concentration of $\bz'$ around the mean correlation. More explicitly, we would like to show that with high probability, the disorder $\bW$ is such that for all the encountered external fields $\bz$, the correlation of the random draw $\bsigma$ is concentrated around the mean correlation $\langle\bsigma\rangle$.
    \item A scalar formula for the mean correlation of a Gibbs sample from an SK model with external field.
\end{itemize}
Both of these points are tractable using spin glass machinery in the high temperature regime for SK, which we describe shortly.
As we shall soon see, the scalar formula that governs the mean correlation is what crucially leads to state evolution equations.

\subsection{The mean recursion and fixed point structure}\label{subsec:amp-connection}
Let us now see more concretely how the fixed point relation arises from the analysis of $\PRGD$.

Define the mean magnetization $\boldm_h = \frac{1}{N} \E_{\bsigma \sim \mu_{\beta \bW , h}} \langle \bsigma,\bone\rangle$, and for $\bsigma \sim \mu_{\beta \bW, h}$, let $\wt{\bg}_h$ be the random variable $\sqrt{N} \cdot \left(\frac{1}{N} \langle \bsigma,\bone\rangle - \boldm_h\right)$.
Under this notation, we can write down the following description of $\PRGD$: from $z \in \R$,
\begin{enumerate}
	\item draw $\bg \sim \calN\left( 0,1 \right)$, set $\bh = \beta\lambda z + \sqrt{\frac{\beta\lambda}{N}} \bg$,
	\item move to $\bz' = \boldm_{\bh} + \frac{1}{\sqrt{N}} \wt{\bg}_{\bh}$. \label{eq:gtilde}
\end{enumerate}

For simplicity of exposition, let us first assume $z = \Omega(1)$; we will comment on the case of small $z$ later.
First, observe by gaussian concentration, with very high probability we have $\bh \approx \beta \lambda z$. 

To simplify the recursion further, we need concentration of $\bz'$ around $\boldm_h$, as well as an explicit formula for $\boldm_h$.
The key observation is that such properties are tractable to prove if the SK model $\mu_{\beta \bW, h}$ is in the high temperature regime.
We first need to introduce some standard spin glass quantities.

\begin{restatable*}[Overlap constant]{definition}{OverlapConstant}
    \label{def:q-overlap}
    For $\beta \ge 0$ and $h > 0$, we define $q = q_{\beta,h}$ as the unique solution to the recursion
    \[ q = \E_{\bg \sim \calN\left(0,1\right)} \left[ \tanh^2 \left( \beta \bg \sqrt{q} + h \right) \right] \mper \]
    We also define $q_{\beta,h} = 0$ for $h = 0$ when $\beta \le 1$, and as the nonzero solution to the above recursion when $\beta > 1$.
\end{restatable*}

The definition above is shown to be well-defined in \Cref{lem:q-properties}.
We now define the high-temperature region for the SK model, a strict weakening of a region first predicted by de Almeida and Thouless \cite{AT78}.

\begin{restatable*}[weak AT line]{definition}{ATLineDef}\label{def:at-line}
A non-negative pair $(\beta,h)$ is said to satisfy the weak Almeida--Thouless ($\textsf{wAT}$) condition, or to be above the weak AT line, if
\begin{equation}
  \label{eq:at-line}
  \tag{\textsf{wAT}}
  \beta^2(1-q_{\beta,h}) = \beta^2 \E \sech^2 \!\left( \beta z\sqrt{q_{\beta,h}} + h \right) < 1 \mper
\end{equation}
\end{restatable*}

We show that if  \eqref{eq:at-line} is satisfied, the correlation $\langle \bsigma,\bone\rangle$ of a Gibbs sample $\bsigma \sim \mu_{\beta \bW, h}$ concentrates around the mean correlation. 
In other words, we have $\bz' \approx \boldm_{\bh}$ with high probability.
Hence, we conclude that a single step of $\PRGD$ can be approximated by the map $z \mapsto \boldm_{\beta \lambda z}$.

Turning now to the mean formula, in \Cref{lem:all-field-mag-conc}, we establish that with high probability over $\bW$, for all $h = O(1)$ where  \eqref{eq:at-line} is satisfied, 
\begin{equation}\label{eq:mean-approx}
    \boldm_{h} \approx \E_{\bg \sim \calN(0, 1)}\left[\tanh(\beta \bg\sqrt{q_{\beta,h}} + h)\right],
\end{equation}
where we recall $q$ from \Cref{def:q-overlap}.
In particular, we establish tighter control on $\boldm_h$ when $h = o(1)$ than is implied by existing literature, which is crucial for analyzing the dynamics near $h = 0$. 

Thus, we have reduced a single step of $\PRGD$ to the deterministic iteration 
\[
    z_{t+1} \gets \E_{\bg \sim \calN(0, 1)} \left[ \tanh\left( \beta\bg\sqrt{q} + \beta \lambda z_t\right) \right]\mcom
\]
where $q = q_{\beta, \beta \lambda z_t}$.

This suggests that the limiting behavior of $\PRGD$ is given by fixed point solutions to the following set of equations:
\begin{equation}
      \label{eq:fixed-point-eq}
      \begin{aligned}
      z &= \E_{\bg \sim \calN\left( 0,1 \right)} \left[ \tanh \left( \beta \bg \sqrt{q} + \beta\lambda z \right) \right] \mcom \\
        q &= \E_{\bg \sim \calN\left( 0,1 \right)} \left[ \tanh^2 \left( \beta \bg \sqrt{q} + \beta\lambda z \right) \right] \mper \\
  \end{aligned}
\end{equation}
Whether $\PRGD$ provably converges to these fixed points is not completely clear, and we defer further discussion on this matter to \Cref{subsec:qual}.

Let us now address the case of small $z$.
When $z = o(1)$, and especially if $z$ is $\Theta(\frac{1}{\sqrt{N}})$, the picture is more delicate, as then the random quantities fluctuate at the same scale as $z$ itself. 
In \Cref{th:rgd-fixed-temp} we use the anticoncentration of $\bg$ to show that in $\Omega(\log N)$ steps, we reach the region $z = \Omega(1)$ with constant probability.

\subsection{RGD simulates AMP}\label{subsec:rgd-amp-simulation}

Let us now draw the explicit connection to AMP fixed points. 
The state evolution fixed points (see \Cref{lem:state-evolution} for a precise statement) that govern the late-stage performance of Bayes AMP for spiked Wigner \cite{DAM16,MV21} read 
\begin{align*}
   \mu &= \lambda \E_{\bg \sim \calN(0, 1)} \qty[\tanh(\lambda \bg \sigma + \lambda^2 \sigma^2)] \\
   \sigma^2  &= \E_{\bg \sim \calN(0, 1)} \qty[\tanh^2(\lambda \bg \sigma + \lambda^2 \sigma^2)]\mper
\end{align*} 
Let $(\mu, \sigma^2)$ be the unique fixed point to the SE equations. We claim that $(z, q) = (\frac{\mu}{\lambda}, \sigma^2)$ is also a solution to \eqref{eq:fixed-point-eq} for $\beta = \lambda$.
Indeed, using that $\E \tanh(\sqrt{\gamma} \bg + \gamma) = \E \tanh^2(\sqrt{\gamma} \bg + \gamma)$ for any $\gamma \ge 0$, we see that $\mu = \lambda \sigma^2$. 
Therefore our putative assignment satisfies $z = q$, and some straightforward algebra yields that $(z, q)$ is a valid solution to the fixed point equations.

As it turns out, the update equations for RGD given by \eqref{eq:fixed-point-eq} are nearly identical to the state evolution equations for mismatched AMP for the spiked Wigner model \cite{AKUZ19}. 
Specifically, we have $\Delta_0 = \frac{1}{\lambda^2}$, $\Delta = \frac{1}{\beta\lambda}$, $\eta_t(\cdot) = \tanh(\cdot)$.
Then some straightforward algebra, combined with \cite[Theorem 1]{AKUZ19}, yields the SE recursion
\begin{align*}
    M_{t+1} &= \E\qty[\tanh(\beta\bg \sqrt{Q_t} + \beta\lambda M_t)] \\
    Q_{t+1} &= \E\qty[\tanh^2(\beta \bg \sqrt{Q_t} + \beta \lambda M_t)]\mper
\end{align*}
The direct comparison is just identifying $M_t \leftarrow z_t$ and $Q_t \leftarrow q_t$. 
The only difference is that $(M_t, Q_t)$ update synchronously, whereas our recursion sets $q_t = q_t(z_t)$ and therefore collapses into an effective one-dimensional recursion on just $z_t$. 
\begin{remark}
    A subtle point here is that even for $\beta = \lambda$, the mismatched AMP is \emph{not} syntactically equivalent to Bayes AMP; in order to derive Bayes AMP, one must artificially enforce that $z = q$.
\end{remark}

\subsection{Qualitative behavior of the fixed point equations}\label{subsec:qual}
The qualitative behavior of the correlation recursion is determined by the fixed points of the recursion, as well as the convergence to said fixed points.

Convergence is nontrivial to guarantee, and indeed the proof of convergence for Bayes AMP does not immediately generalize to the mismatched setting where $\beta \neq \lambda$.
To this end, in \Cref{lem:fixed-point} we establish that so long as $(\beta, \beta \lambda z)$ satisfies \eqref{eq:at-line}, the map $z \mapsto \boldm_{\beta \lambda z}$ is strictly increasing in $z$. 
For any $\beta < 1$, \eqref{eq:at-line} is automatically satisfied, and in this setting we are able to fully characterize convergence for the recursion starting from any $z$. 

Turning now to the fixed points, one can immediately identify $z = 0$ as one fixed point of this recursion. As it turns out, this map $z \mapsto \boldm_{\beta \lambda z}$ behaves differently depending on $\beta$, as illustrated in \Cref{fig:magnetization-predictions}:
\begin{enumerate}[label=(\roman*)]
	\item If $\beta$ is too small (smaller than $\frac{1}{\lambda}$), the signal is washed out and $0$ is the only fixed point of this recursion. Furthermore, it is a stable fixed point.
	\item If $\beta$ is of intermediate magnitude (between $\frac{1}{\lambda}$ and $1$), in addition to an unstable fixed point at $0$, this function has stable fixed points at $\pm \OPT_{\beta,\lambda}$.
	\item If $\beta > 1$, the formula for the mean correlation $\boldm_{\beta\lambda z}$ is only valid if $z$ is sufficiently large. Indeed, if $z$ is close to $0$, \eqref{eq:at-line} is violated, and the corresponding SK model is known to be in the low-temperature replica-symmetry breaking (RSB) regime. 
    Hence, the recursion is only valid for large $z$, and in this region there are at most two stable fixed points $\pm \OPT_{\beta,\lambda}$.
	\item Finally, if $\beta$ is very large, the RSB region for $z$ is so large that it swallows the fixed points that are present at smaller $\beta$.
\end{enumerate}

\begin{figure}[!ht]
  \centering

  \begin{subfigure}[t]{0.45\textwidth}
    \centering
    \includegraphics[width=\linewidth]{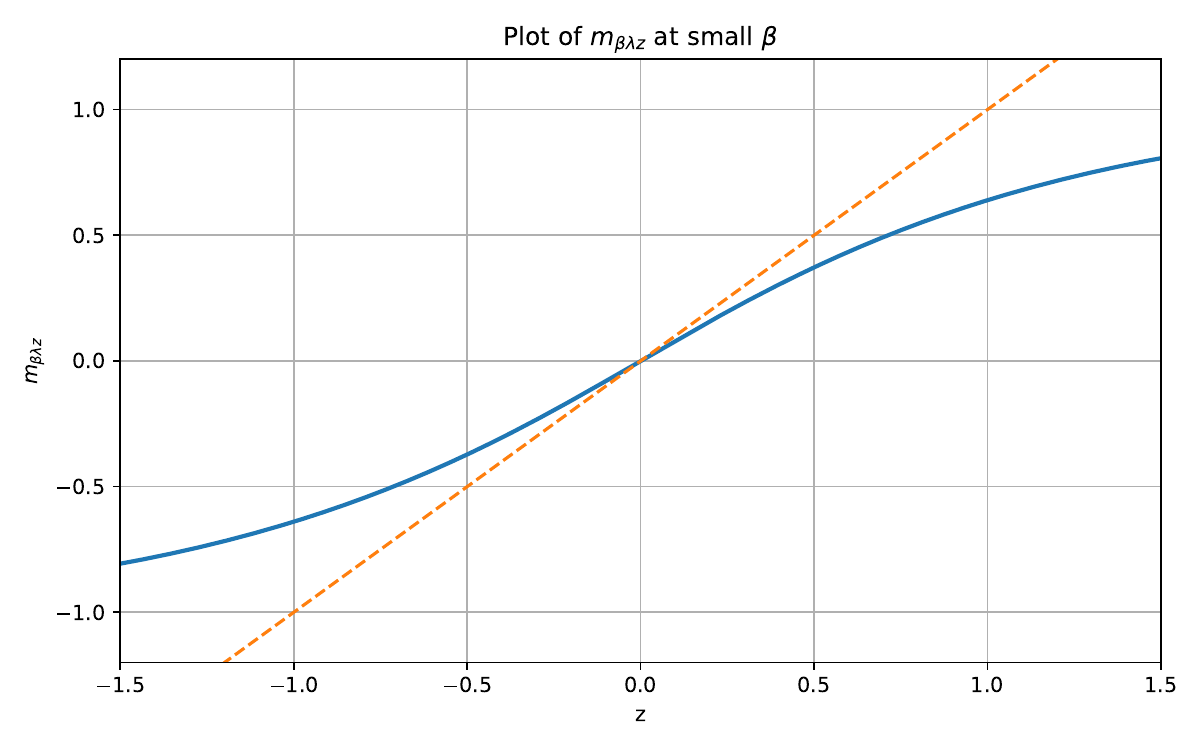}
    \par\vspace{0.35ex}
    \parbox[t][\capht][t]{\linewidth}{%
      \caption{$\boldm_{\beta\lambda z}$ at $\beta < \frac{1}{\lambda}$: $0$ is the only fixed point.}%
    }
  \end{subfigure}
  \hfill
  \begin{subfigure}[t]{0.45\textwidth}
    \centering
    \includegraphics[width=\linewidth]{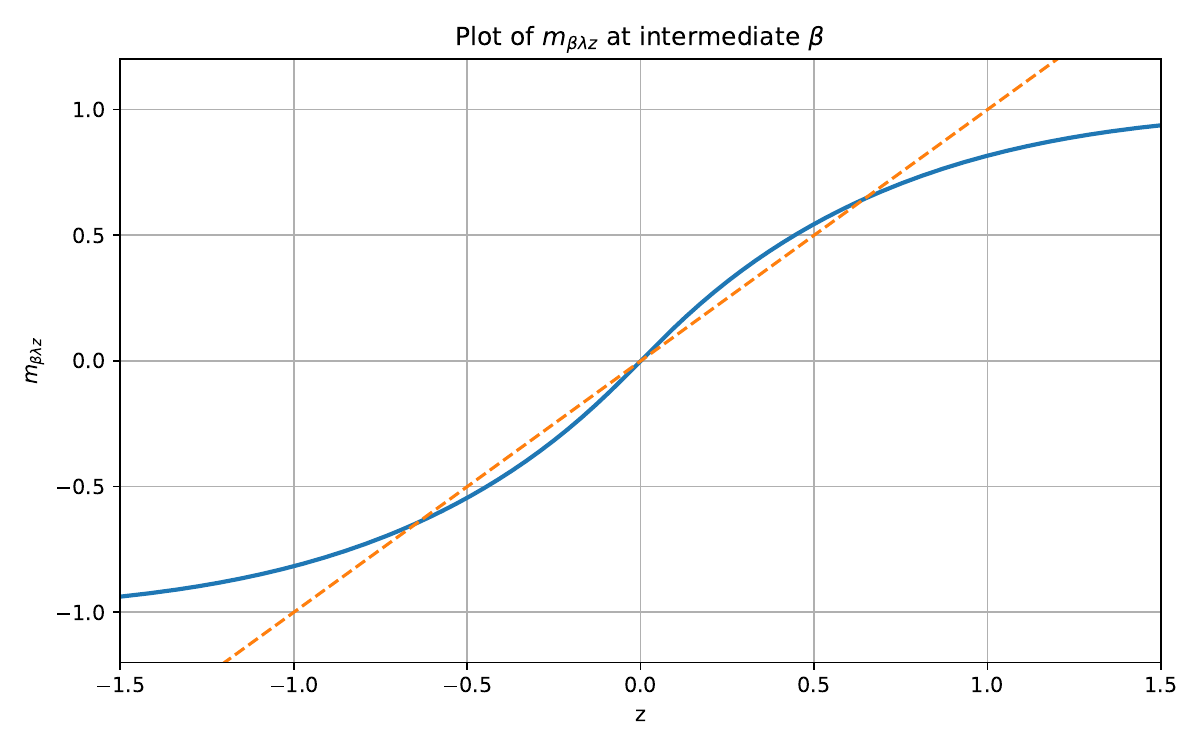}
    \par\vspace{0.35ex}
    \parbox[t][\capht][t]{\linewidth}{%
      \caption{$\boldm_{\beta\lambda z}$ at intermediate $\frac{1}{\lambda} < \beta < 1$: non-trivial fixed points appear.}%
    }
  \end{subfigure}

  \vspace{1em}

  \begin{subfigure}[t]{0.45\textwidth}
    \centering
    \includegraphics[width=\linewidth]{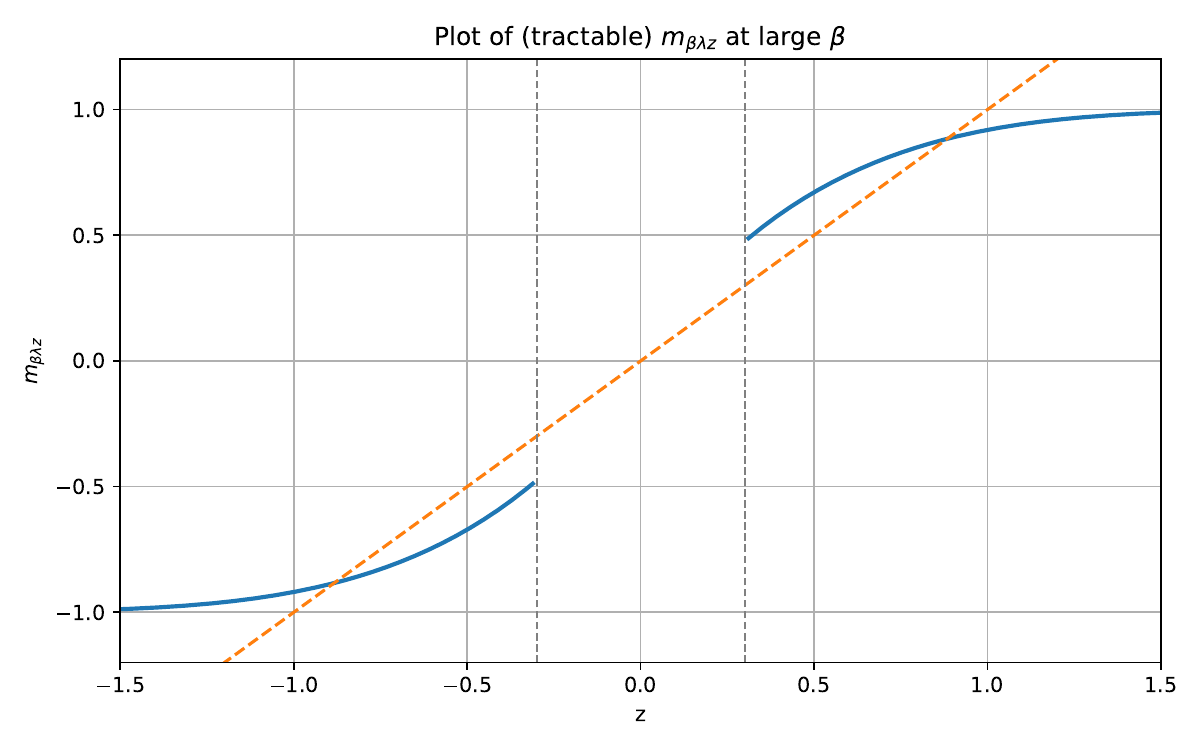}
    \par\vspace{0.35ex}
    \parbox[t][\capht][t]{\linewidth}{%
      \caption{$\boldm_{\beta\lambda z}$ at $1 < \beta < \lambda$: the $0$ fixed point is swallowed by intractability.}%
    }
  \end{subfigure}
  \hfill
  \begin{subfigure}[t]{0.45\textwidth}
    \centering
    \includegraphics[width=\linewidth]{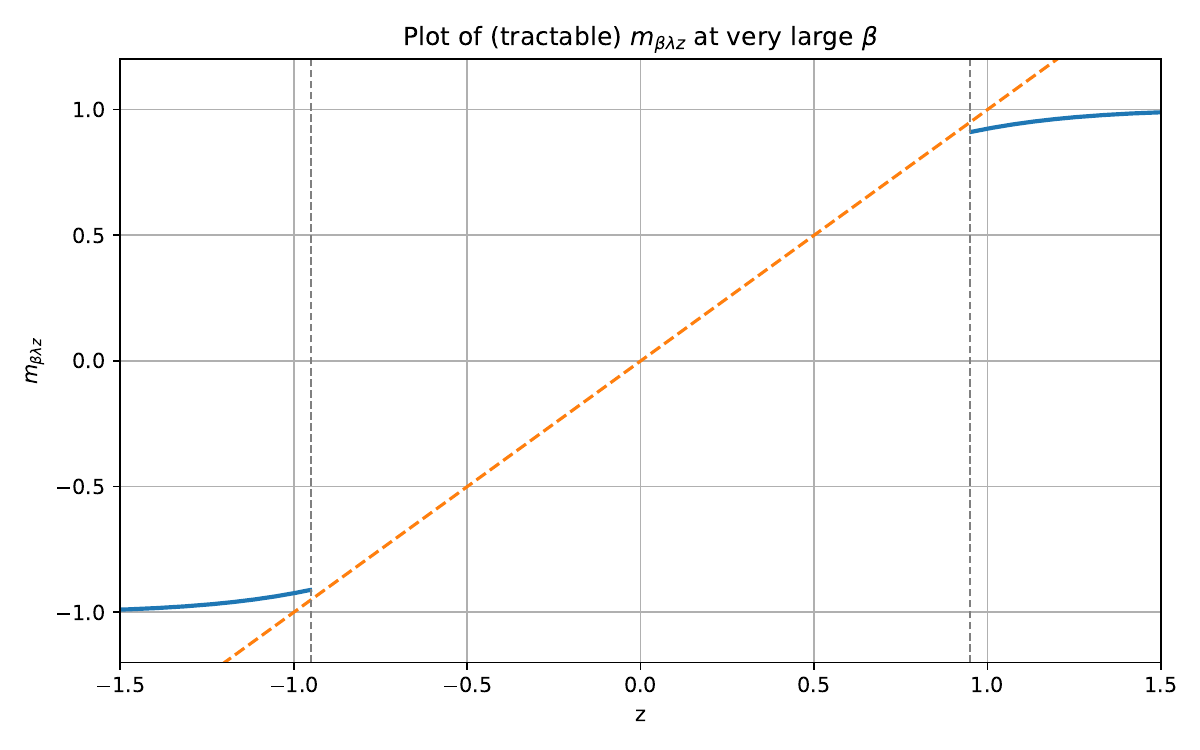}
    \par\vspace{0.35ex}
    \parbox[t][\capht][t]{\linewidth}{%
      \caption{$\boldm_{\beta\lambda z}$ at $\beta \gg \lambda$: all fixed points are swallowed by intractability.}%
    }
  \end{subfigure}

  \caption{Behavior of the RGD recursion for different values of $\beta$.}
  \label{fig:magnetization-predictions}
\end{figure}

We make the above predictions for the convergence and fixed points rigorous in \Cref{lem:fixed-point}.
\Cref{fig:magnetization-predictions} suggests two kinds of results one could hope to show.

\begin{enumerate}[label=(\alph*)]
    \item First, that for any $\beta$, the RGD Markov chain reaches the stable $\OPT_{\beta, \lambda}$ fixed point from a sufficiently warm start which puts $z_0$ in the high-temperature, replica-symmetric (RS) regime. That is, if the initialization $z_0$ of $\PRGD$ is such that $(\beta, \beta \lambda z_0)$ satisfies \eqref{eq:at-line}, then the correlation shoots to the stable fixed point $\OPT_{\beta, \lambda}$ in a small constant number of steps of RGD. This is essentially immediate from classical analyses of fixed point recursions; see \Cref{lem:rgd-convergence-from-warm-start} for details.

  \item Second, that for $\beta < 1$, the RGD Markov chain reaches the stable $\OPT_{\beta, \lambda}$ fixed point from an \emph{arbitrary} initialization. In the case that $\beta < \frac{1}{\lambda}$, $0$ is the only stable fixed point, so for now we assume that $\beta \in (\frac{1}{\lambda}, 1)$. Since the previous point implies that the correlation goes to $\OPT_{\beta, \lambda}$ once it becomes sufficiently large, the interesting part here is showing that the Markov chain successfully escapes the unstable fixed point at $0$. We establish this in \Cref{th:rgd-cold-start-analysis}, where we show that the noise $\bg$ and $\wt{\bg}$, which are typically $O\left( \frac{1}{\sqrt{N}} \right)$, are anticoncentrated enough to cause the correlation to escape $0$.
\end{enumerate}
\begin{remark}
    As a simple illustrative example, we can carry out the same analysis on the RGD dynamics for the Curie-Weiss model at inverse temperature $\beta$ and external field $h$. In this case $\boldm_{\bh}$ can be computed exactly as $\tanh(\bh)$, and so we see that RGD converges to the stable fixed point of the mean-field recursion $m = \tanh(\beta m + h)$.
\end{remark}

\section{Conjectures and open questions}\label{sec:open}

One natural question is whether we can extend our analysis to other inference problems.
At least for the case of spiked matrix problems, the general recipe seems fairly adaptable. 
For example, consider a more general spiked matrix inference problem with a different prior and noise model (i.e. not $\GOE$).
The main obstacles here are (1) establishing replica symmetry (specifically overlap concentration) for this more general spin glass, and using it to show sample magnetization concentration \`{a} la \Cref{lem:strong-field-magnetization} and (2) proving a mean correlation formula, analogous to \Cref{lem:all-field-mag-conc}.

\begin{question}
    Can one establish a kind of universality result for our fixed point characterization that does not rely on model-specific calculations?
\end{question}

It would also be interesting to see if one could push the limits of our analysis to the multi-spike setting and achieve results analogous to \cite{BGP24} for Glauber dynamics.

Another interesting (but open-ended) question is: what other classes of inference problems can be analyzed using locally stationary distributions? 
The success of RGD in the present analysis is rather contingent on the spiked matrix structure; in other inference problems, it is unclear what Markov chain should be used as a proxy to understand Glauber dynamics.

\begin{question}
    Identify broader inference problems that can be analyzed through the lens of locally stationary distributions.
\end{question}

We next turn to open questions regarding the spiked Wigner problem in particular. 
One fundamental barrier to applying our techniques to understand Glauber for $\beta > 1$ is that the transfer of local stationarity (\Cref{lem:lsd-transfer}) does not work in this setting, even if we assume mixing guarantees in the entire conjectured high-temperature region \eqref{eq:at-line}. In particular, it is unclear how to prove a version of \Cref{lem:lsd-transfer} where we only assume that the component measure satisfies an MLSI with high probability, under the assumption that the initialization from which the locally stationary distribution is obtained lives in this high-temperature region.

\begin{question}\label{conj:glauber-full-range}
    Assuming SK mixes fast for all $(\beta, h)$ satisfying \eqref{eq:at-line}, can we extend our analysis to characterize the recovery performance for annealed Glauber dynamics up to $\beta \in \left( 1, \beta_{\lambda} \right]$?
\end{question}

It seems likely that if the above question were answered affirmatively, this would provide an alternate algorithm that attains Bayes-optimal performance in the entire high-SNR regime $\lambda > 1$.

We note that \cite{AKUZ19} makes some predictions for the behavior of these models in the RSB regime, i.e. beyond $\beta_{\lambda}$.  

For technical reasons, the current techniques which involve a transfer of local stationarity between Glauber and RGD necessarily incur some $\poly(N)$ overhead in the runtime. 
On the other hand, based on physics heuristics and the rigorous results of \cite{Sel24,DGPZ25}, one would expect the dynamics to achieve nonzero correlation in even $O(N)$ steps. 
\begin{conjecture}\label{conj:glauber-linear}
    Glauber dynamics succeeds at weak recovery in (near-)linear runtime.
\end{conjecture}

One path towards achieving such a guarantee is to establish exponential contraction of the local stationarity parameter, perhaps by establishing a weak \Poincare inequality \cite{HMRW24} for symmetric test functions.
Indeed, we believe that our techniques to analyze RGD might help in showing that annealed Glauber dynamics succeeds at sampling from the posterior (and potentially even from lower-temperature models). Furthermore, while the ``conservation of variance'' property used to show sampling guarantees for annealed Glauber dynamics is typically reliant on proving certain covariance bounds (e.g. \cite[Lemma 6.8]{HMRW24}) which are often technically challenging, perhaps this can be sidestepped if one instead directly tried to establish a weak \Poincare inequality for the (one-dimensional!) RGD Markov chain.

\begin{conjecture}\label{conj:glauber-mixing}
    There is some threshold $\beta_{\lambda} > \lambda$ such that for all $\beta < \beta_{\lambda}$, annealed Glauber dynamics samples from the scaled posterior $\mu_{\beta \bM}$ in $\poly(N)$, or even  $O(N\log N)$ time. Here, define $\beta_{\lambda}$ as the largest temperature $\beta$ at which there is a nonzero fixed point $z$ to the recursion \eqref{eq:fixed-point-eq} such that $(\beta,\beta\lambda z)$ is above the AT line $\beta^2 \E\left[ \sech^4(\beta\bg\sqrt{q} + \beta\lambda z) \right] < 1$.
\end{conjecture}




\section{Preliminaries}
\label{sec:prelims}

\subsection{Notation}

\begin{itemize}
    \item Given $\sigma_1,\sigma_2 \in \{\pm 1\}^{N}$, we denote $R(\sigma_1,\sigma_2) = \frac{1}{N} \langle\sigma_1,\sigma_2\rangle$ to be the normalized correlation between the two.
    \item We use $c$ to denote small positive constants whose value may change from line to line, and $C$ to denote similarly fickle large constants.
\end{itemize}

\subsection{Markov chain basics}

\begin{definition}[Ising model]
    Let $J \in \R^{n \times n}$ be a symmetric \emph{interaction matrix} and $h \in \R^n$ an \emph{external field}.
    The \emph{Ising model} corresponding to $J$ and $h$ is the probability distribution $\mu_{J, h}$ on $\{\pm1\}^n$, where
    \[
        \mu_{J, h}(x) \propto \exp\parens*{\frac{1}{2}x^\top J x + \langle h,x\rangle}.
    \]
    We drop the $h$ from the subscript when it is equal to $0$.
\end{definition}

We also recall the standard definitions of functional inequalities which imply rapid mixing.
For a reversible Markov chain $P$ on state space $\Omega$ with stationary distribution $\pi$, and 
any functions $f,g:\Omega\to\R$, define the \emph{Dirichlet form} $\calE(f,g) \coloneqq \E_{\bx\sim\pi} \E_{\by\sim_{P}\bx} (f(\bx)-f(\by))(g(\bx)-g(\by))$.
\begin{definition}[\Poincare inequality]
    We say $P$ satisfies a \emph{\Poincare inequality} with constant $C$ if for any function $f:\Omega\to\R$:
    \[
        \calE(f,f) \ge C\cdot\Var[f]\mper
    \]
\end{definition}

\begin{definition}[Modified log-Sobolev inequality]\label{def:mlsi}
    We say
    $P$ satisfies
    a \emph{modified log-Sobolev inequality} (MLSI) with constant $C$ if for any function $f : \Omega \to \R_{> 0}$,
        \[ \mathcal{E}_{P}(f,\log f) \ge C \cdot \Ent[f]\mcom\]
    where $\Ent[f] \coloneqq \E_{\pi}[f \log f] - \E_{\pi} f \log \E_{\pi} f$.
\end{definition}

\subsection{Locally stationary distributions}
We now give the relevant definitions for locally stationary distributions. 
\begin{definition}[{\cite[Definition 1.2]{LMRRW24}}]
    Let $P$ be a time-reversible ergodic Markov chain on a finite state space $\Omega$, with associated stationary distribution $\pi$. A probability measure $\nu$ on $\Omega$ is said to be \emph{$\eps$-locally stationary} with respect to $P$ if, with $f = \dv{\nu}{\pi}$,
    \[ \calE\left( f , \log f \right) \defeq \sum_{x,y \in \Omega} P_{x \to y} \left( f(x) - f(y) \right) \cdot \log \frac{f(x)}{f(y)} \le \eps \mper \]
\end{definition}

One can also show that at a typical time $t$, \emph{any} reversible Markov chain yields a locally stationary distribution.
\begin{lemma}
    \label{lem:lsd-decay}
    Let $\nu_0$ be an arbitrary distribution, $\nu_t = P^t \nu_0$, and set $f_t = \frac{\dif \nu_t}{\dif \pi}$. Also set $f = \E_{\bt \sim [0,T]} f_{\bt}$, and $\nu \defeq \E_{\bt \sim [0,T]} \nu_{\bt}$. Then,
    \[ \calE \left( f , \log f \right) \le \frac{\KL{\nu_0}{\pi}}{T} \le \frac{\log(1/\pi_{\mathrm{min}})}{T} \mper \]
    In particular, $\nu$ is $\frac{\log(1/\pi_{\mathrm{min}})}{T}$-locally stationary with respect to $\pi$.
\end{lemma}
\begin{proof}
    This immediately follows from \cite[Lemma 3.1]{LMRRW24} (also see \cite[Lemma II.1]{BCV25}), coupled with the observation that the Dirichlet form $f \mapsto \calE(f, \log f)$ is convex.
\end{proof}

We will also need to use the fact that a locally stationary distribution does not change much in TV distance if we apply $P$ for a few steps.

\begin{lemma}[{TV stability, essentially \cite[Lemma 3.2]{LMRRW24}}]
    \label{lem:dirichlet-tv-relation}
    Let $P$ be a reversible Markov chain with stationary distribution $\pi$,  $\nu$ an arbitrary distribution which is $\eps$-locally stationary. 
    Then,
    \[ 
    \dtv{\nu}{P^T\nu} \le \frac{T\sqrt{\eps}}{2}. \]
\end{lemma}
\begin{proof}
    When $T=1$, this is \cite[Lemma 3.2]{LMRRW24}. 
    The data processing inequality implies that for any $k \ge 1$, $\dtv{P^{k} \nu}{P^{k+1}\nu} \le \frac{\sqrt{\eps}}{2}$. 
    The claim follows on applying the triangle inequality.
\end{proof}

\begin{corollary}  \label{lem:bounded-function-stability}
    Let $\phi: \Omega \to \R$ be a bounded function on the state space of a Markov chain $P$, and $\nu$ an $\eps$-locally stationary measure.  Then,
    \[
        | \E_{\bx \sim \nu}[ \phi(\bx)] - \E_{\bx \sim P^T \nu} [ \phi(\bx)] | \leq \| \phi \|_{\infty} \cdot T\sqrt{\eps}.
    \]
\end{corollary}

We will also need to transfer local stationarity between Glauber and RGD. 
\begin{lemma}[{\cite[Lemma 3.7]{LMRRW24}}]\label{lem:lsd-transfer}
    Let $P$ be the Markov chain associated to a measure decomposition $\pi = \E_{\bz \sim \rho} \pi_{\bz}$. 
    Let $f: \{\pm1\}^n \to \R_{>0}$ be any function and set $\tau$ such that $\min_{x \in \{\pm 1\}^n} f(x) > \exp(-\tau)$ or $\max_{x \in \{\pm 1\}^n} f(x) < \exp(\tau)$.
    For $\delta \coloneqq \inf_{z} \MLSI\parens*{\pi_{z}}$, we have
    \[
        \calE_{P}(f, \log f) \le O\parens*{\frac{\tau}{\delta}} \cdot \calE_{\pi}(f,\log f)\mper
    \]
\end{lemma}

\subsection{Approximate Message Passing and State Evolution}
For simplicity, we only recall the specific version of AMP used for spiked Wigner \cite{DAM16}.
\begin{definition}
    For $t \ge 0$, update 
    \begin{align*}
        m^t &= \tanh(x^t)\mcom \\
        \mathsf{b}_t &= \frac{\lambda^2}{N} \sum_{i=1}^{N} (1 - \tanh^2(x^t_i)) \\
        x^{t+1} &= \lambda \bM m^t - \mathsf{b}_t m^{t-1} 
    \end{align*}
    where $\tanh(\cdot)$ is applied entrywise.
\end{definition}

Above, $m^{-1} = x^0$ may be initialized according to one of a variety of schemes, such as a spectral initialization, e.g. \cite{MV21}, or a random initialization, e.g. \cite{LFW23}. State evolution \cite{BM11} predicts the empirical distribution of the coordinates of $x^t$. 
We record the special case relevant to our setting: the correlation of the AMP iterate with the planted signal $\bx$ can be predicted by a scalar recursion.
\begin{lemma}[State evolution for correlation]\label{lem:state-evolution}
    Fix a time $t \ge 0$. Let $(\mu_t, \sigma_t)$ be governed by the recursion:
    \begin{align*}
        \mu_{t+1} &= \lambda \E_{\bg \sim \calN(0, 1)}[\tanh(\lambda \bg \sigma_t + \lambda^2\sigma_t^2)] \\
        \sigma_{t+1}^2 &= \E_{\bg \sim \calN(0, 1)}[\tanh^2(\lambda \bg \sigma_t + \lambda^2\sigma_t^2)]\mper
    \end{align*}
    Let $X_0 \sim \mathrm{Unif}\{\pm 1\}, Z_0 \sim \calN(0, 1)$. As $N \to \infty$, we have the following convergence in $W_2$:
    \[
        \frac{1}{N} \ev{m^{t}, \bx} \to \E[X_0(\mu_t X_0 + \sigma_t Z_0)]\mper
    \]
    
\end{lemma}

\subsection{The Sherrington--Kirkpatrick model}

In this subsection, we introduce several important parameters in the SK model. 

\OverlapConstant

It is known that in the high-temperature replica-symmetric phase $\beta < 1$, the overlap between two samples from the corresponding SK model concentrates around the above number.

\ATLineDef

Observe that if $\beta < 1$, $(\beta,h)$ satisfies \eqref{eq:at-line} for any $h \in \R_{\ge 0}$.
A very important question in the study of the SK model is when the model is replica-symmetric: when does the overlap between two independent samples from the Gibbs measure concentrate?

\begin{definition}
    We say that $(\beta,h)$ satisfies \emph{exponential overlap concentration} if for some constant $C$ independent of $N$,
    \begin{equation}
        \label{eq:overlap-conc}
        \tag{\textsf{overlap-conc}}
        \E_{\substack{\bW \sim \GOE(N) \\ \bsigma^1,\bsigma^2 \sim \mu_{\beta\bW,h\bone}}} \left[ \exp \left( \frac{N \left( \frac{1}{N} \langle \bsigma^1,\bsigma^2 \rangle - q_{\beta,h} \right)^2}{C} \right) \right] \le 2 \mper
    \end{equation}
\end{definition}

It is now well-known that above the weak AT line, overlap concentration holds.

\begin{lemma}
    \label{lem:wat-overlap-conc}
    Let $(\beta,h)$ satisfy \eqref{eq:at-line}. Then, $(\beta,h)$ satisfies \eqref{eq:overlap-conc}.
\end{lemma}


\begin{proof}
    The above follows essentially immediately from \cite[Lemma 6.1 and Section 8.3]{JT17} and \cite[Lemma 13.3.4 and Theorem 13.7.1]{Tal11}; we omit the details.
\end{proof}

We will also need the following simple fact.

\begin{restatable}{lemma}{watmonotonicity}
    \label{lem:at-monotonicity}
    \label{cor:wat-continuous}
    Let $(\beta,h)$ satisfy \eqref{eq:at-line}. Then, $(\beta,h')$ satisfies \eqref{eq:at-line} for all $h' > h$.
\end{restatable}


\begin{remark}
    Recently, Lopatto \cite{Lop26} established rigorously that replica-symmetry holds when $(\beta,h)$ is above the (non-weak) Almeida--Thouless line $\beta^2 \E\left[ \sech^4 \left( \beta\bg\sqrt{q} + h \right) \right] < 1$ \cite{AT78}, closing a long-standing open problem. It is further believed that replica-symmetry more generally implies exponential overlap concentration in the SK model; we do not attempt to establish this, and focus on the simpler setting of the weak AT line as above, which will suffice for our purposes. On the other hand, Toninelli \cite{Ton02} showed that if $(\beta,h)$ does \emph{not} satisfy the AT condition, the corresponding model in fact exhibits replica-symmetry breaking.
\end{remark}



\section{The dynamics of RGD}
\label{sec:rgd-dynamics}

The goal of this section will be to understand the trajectory of RGD---we will establish that RGD rapidly converges to fixed points of the one-dimensional recursion of interest. 
Recall from the technical overview that for $\lambda > 1$ and $\beta > 0$, this recursion is given by the update function $f_{\beta}$ defined by
\begin{equation}
    \label{eq:f-beta-lambda}
    f_{\beta,\lambda}(z) = \E\left[ \tanh\left( \beta\bg\sqrt{q_{\beta,\beta\lambda z}} + \beta\lambda z \right) \right] \mper
\end{equation}

When $\lambda$ is clear from context, we refer to $f_{\beta, \lambda}$ as simply $f_{\beta}$. The fixed point(s) of the above recursion will be an object of fundamental interest in this section.

\begin{lemma}
    \label{lem:at-fixed-pt}
    For $\lambda > 1$ and $\frac{1}{\lambda} < \beta \le \lambda$, there exists a unique positive fixed point $\OPT_{\beta,\lambda}$ of $f_{\beta,\lambda}$. Furthermore, it is a stable fixed point and $(\beta,\beta\lambda\OPT_{\beta,\lambda})$ satisfies \eqref{eq:at-line}.
\end{lemma}

Let us now state the main result for this section.

\begin{restatable}[Fixed-temperature RGD]{theorem}{thrgdfixedtemp}
	\label{th:rgd-fixed-temp}
	Fix $\lambda > 1$, $K > 0$ large, $\eps > 0$ small, and $\frac{1}{\lambda} < \beta \le \lambda$. With high probability over the noise $\bW$, the following holds. Let $x_0 \in \{\pm 1\}^n$, and assume that $( \beta , \beta\lambda R(x_0,\bone) )$ satisfies \eqref{eq:at-line}. Suppose that we run the RGD Markov chain at inverse temperature $\beta$ from $x_0$ for $T \ge \omega(\log N)$ steps to arrive at distribution $\nu_T$. Suppose $T \le N^{K}$. Then,
	\[ \E_{\bx_T \sim \nu_T}\left[ \left| |R\left( \bx_T,\bone \right)| - \OPT_{\beta,\lambda} \right| \right] \le O\qty(\frac{1}{N^{1/2 - \eps}}) \mper \]
\end{restatable}



	

Combining \Cref{th:rgd-fixed-temp} with the properties of locally stationary distributions, we obtain the following control over the recovery performance of Glauber when $\beta < 1$.

\begin{theorem}[Fixed-temperature Glauber]\label{thm:glauber-formal}
    Let $\beta \in (\frac{1}{\lambda}, 1)$ satisfy \condref{cond:sk-beta-one} and $\eps > 0$. With high probability over $\bW$, the following holds. Let $\bsigma_0 \in \{ \pm 1 \}^N$ be arbitrary, and let $(\bsigma_t)_{t \ge 0}$ be the trajectory of the Glauber dynamics run with stationary distribution $\mu_{\beta\bM}$ initialized at $\bsigma_0$. Let $T \ge \wt{\Omega}(N^{4 - 2\eps})$, and $\wh{\bsigma} = \bsigma_{\bt}$ for a uniformly random $\bt \sim [0,T]$. Then, there exists an explicit constant $\OPT_{\beta,\lambda} > 0$ such that
   \[ \Pr \left[ \left| \frac{1}{N} \left| \langle \wh{\bsigma} , \bx \rangle  \right| - \OPT_{\beta,\lambda} \right| > \frac{C}{N^{1/2 - \eps}} \right] = o_N(1) \mper \]
\end{theorem}
\begin{proof}
Let us condition on the very high probability event that $\opnorm{\bW} = O(1)$. Let $\nu$ denote the distribution $\E_{\bt \sim [0, T]} \nu_{\bt}$, and $f$ its relative density with respect to $\mu_{\beta \bM}$. By \Cref{lem:lsd-decay}, we have that $\calE(f , \log f) \le O\qty(\frac{N}{T})$.
Applying \Cref{lem:lsd-transfer} yields that $\calE_{\mathrm{RGD}}(f_{\bt}, \log f_{\bt}) \le O\qty(\frac{\tau}{\delta} \cdot \frac{N}{T})$. 
We can crudely bound $\tau = O(N)$ because $\opnorm{\beta \bM} = O(1)$, and since we assumed that $\beta$ satisfies \condref{cond:sk-beta-one}, we have $\delta = \Omega(\frac{1}{N})$, yielding that 
\[
\calE_{\mathrm{RGD}}(f_{\bt}, \log f_{\bt}) \le O\qty(\frac{N^3}{T})\mper
\]

We can now apply \Cref{lem:bounded-function-stability} to $\nu$ with the test function $\phi(x) = \abs{\left| R(x, \bone) \right| - \OPT_{\beta, \lambda}}$, to conclude that
\[
\abs{\E_{\nu}\qty[\abs{R(\bx, \bone) - \OPT_{\beta, \lambda}}] - \E_{(\PRGD)^{C\log N}\nu}\qty[\abs{R(\bx, \bone) - \OPT_{\beta, \lambda}}] } \le \wt{O}\qty(\sqrt{\frac{N^3}{T}}).
\]
Now since $\beta < 1$, \eqref{eq:at-line} automatically holds. Therefore, \Cref{th:rgd-fixed-temp} yields the appropriate control on the latter expectation, and we conclude by applying Markov.
\end{proof}

We shall prove \Cref{th:rgd-fixed-temp} in two steps. First, in \Cref{lem:rgd-convergence-from-warm-start}, we shall prove the claim assuming a warm start initialization where $|R(x_0,\bone)| = \Omega(1)$; such a warm start is necessary for $\beta > 1$ since the AT condition is nontrivial there. We then show in \Cref{th:rgd-cold-start-analysis} that when $\beta < 1$, RGD escapes the unstable fixed point where $R(x_0,\bone) = 0$.

For use later in this section, we also note the following about the structure of $f_{\beta,\lambda}$, a more precise formulation of \Cref{eq:f-beta-lambda} proved in \Cref{sec:understanding-fixed-points}.

\begin{restatable}{lemma}{lemfixedptstructure}
	\label{lem:fixed-pt-structure-unconditional}
    Fix $\lambda > 1$ and $\beta \ge 0$. For $f_{\beta} = f_{\beta,\lambda}$ defined as in \eqref{eq:f-beta-lambda},
    \begin{enumerate}[label=\normalfont{(\roman*)}]
        \item If $\beta < \frac{1}{\lambda}$, then $h=0$ is the only fixed point for $f_{\beta}(h)$. Furthermore, $h=0$ is a stable fixed point.
        \item If $\frac{1}{\lambda} < \beta \le \lambda$, then there are exactly two non-negative fixed points for $f_{\beta}$, one at $0$ and another at $\OPT_{\beta,\lambda}$. Furthermore, $h=0$ is an unstable fixed point and $h=\OPT_{\beta,\lambda}$ is a stable fixed point.
    \end{enumerate}
\end{restatable}


\subsection{RGD from a warm start reaches the stable fixed point}

	In this section, we will show that running RGD on the posterior of spiked Wigner attains the Gibbs correlation, provided that the Markov chain is initialized at a point (or distribution) with typically large correlation. Our analysis will require RGD to be run for at least a sufficiently large constant number of steps, but the analysis extends to any large polynomial number of steps. In contrast, existing nonasymptotic analyses of AMP only work for sublinearly many steps.
	
	\begin{lemma}
        \label{lem:rgd-convergence-from-warm-start}
		Fix $\lambda > 1$, $K > 0$ large, $\eps > 0$ small, and $\frac{1}{\lambda} < \beta \le \lambda$. With high probability over the noise $\bW$, the following holds. Let $x_0 \in \{\pm 1\}^n$, and assume that $( \beta , \beta\lambda |R(x_0,\bone)| )$ satisfies \eqref{eq:at-line}, and $|R(x_0,\bone)| = \Omega(1)$.\footnote{Here and below, satisfying \eqref{eq:at-line} means we satisfy the weak AT line inequality with some constant gap.} Suppose that we run the RGD Markov chain at inverse temperature $\beta$ from $x_0$ for $T \ge \omega(\log N)$ steps to arrive at distribution $\nu_T$. Suppose $T \le N^{K}$. Then,
    	\[ \E_{\bx_T \sim \nu_T}\left[ \left| |R\left( \bx_T,\bone \right)| - \OPT_{\beta,\lambda} \right| \right] \le O\qty(\frac{1}{N^{1/2 - \eps}}) \mper \]
	\end{lemma}

	\begin{remark}
		If one only desired to have correlation within a(n arbitrarily small) constant of $\OPT_{\beta}$, our proof (specifically \eqref{eq:one-dim-fixed-point-convergence}) will show that it suffices to run RGD for a (large) constant number of steps. 
	\end{remark}

	\begin{proof}
		For ease of notation, let $c > 0$ be a sufficiently small constant such that $(\beta, \beta\lambda c)$ satisfies \eqref{eq:at-line}, and let 
		\[ \calH = \left\{ h \in [-2 , 2] : (\beta,\beta\lambda h) \text{ satisfies \eqref{eq:at-line} and $|h| > c$} \right\} \mper \]
        Observe that $\calH$ contains $R(x_0, \bone)$.       
		Consider the one-dimensional projection of the RGD walk, initialized at $z_0 = R(x_0,\bone)$. One step of this one-dimensional walk is defined by walking from $\bz_t$ as
		\[ \bz_{t+1} = R(\bx_{t+1} , \bone) \mcom \text{ where } \bx_{t+1} \sim \mu_{\beta W , \beta\lambda \bz_t + \sqrt{\frac{\beta\lambda}{n}} \bg_t }\mcom \]
		where each $\bg_t$ is iid drawn from $\calN(0,1)$. Let us condition on the following high-probability events:
		\begin{enumerate}[label=(\roman*)]
			\item For every $t < T$, $|\bg_t| < N^{\eps}$. By standard Gaussian tail bounds, this occurs with probability $1 - T \cdot e^{-\Omega(N^{\eps})} = 1 - e^{-cN^{\eps}}$.

			\item For all $h \in \calH$ and constant $k$,
            \[ \E_{\bsigma \sim \mu_{\beta\bW , \beta\lambda h}} \left( R(\bsigma,\bone) - \E\left[ \tanh\left( \beta\sqrt{q_{\beta,\beta\lambda h}} \bg + \beta\lambda h \right) \right] \right)^{2k} < \frac{1}{N^k} \mper \]
            This follows from \Cref{lem:strong-field-magnetization} and \Cref{lem:wat-overlap-conc}, on taking a polynomially fine net of $\calH$.
		\end{enumerate}

        We claim that, on the above events, for all $t < T$, $\beta \lambda \bz_{t}$ is bounded away from the AT line by a constant.
        The base case $t=0$ is true by assumption.
        Next, since $\beta\lambda \bz_t$ is bounded away from the AT line by a constant, (i) implies that $\beta \lambda \bz_t + \sqrt{\frac{\beta\lambda}{n}} \bg_t$ does so as well. (ii) then implies that with high probability, for all $t < T$,
		\[
			\left| \bz_{t+1} - f_{\beta}\left( \bz_t + \sqrt{\frac{1}{\beta\lambda N}} \bg_t \right) \right| \le O\left(\frac{1}{N^{1/2 - \eps}}\right)\mper
		\]
		The Lipschitzness of $f_{\beta,\lambda}$ with (i) implies that
		\begin{equation}
			\label{eq:dynamic-growth-one-dim}
			\left| \bz_{t+1} - f_{\beta}(\bz_t) \right| \le O\left(\frac{1}{N^{1/2 - \eps}}\right)\mper
		\end{equation}
		Note in particular that by the monotonicity of \Cref{lem:at-monotonicity}, if the above is true, then $\beta\lambda\bz_{t+1}$ also satisfies \eqref{eq:at-line} with a constant gap. Indeed, by \Cref{lem:fixed-pt-structure-unconditional}(ii), if $|z_t| < \OPT_{\beta,\lambda}$, we have $|f_{\beta}(z_t)| > |z_t|$, so it satisfies \eqref{eq:at-line}. 
        On the other hand, if $|z_t| > \OPT_{\beta,\lambda}$, then $|f_{\beta}(z_t)| > \OPT_{\beta,\lambda}$, so it satisfies \eqref{eq:at-line} due to \Cref{lem:at-monotonicity}.
		As a result, we conclude that for all $t < T$, using $f_\beta$ is valid as an approximation to the RGD update at step $t$.

        Next, we show that $|z_t|$ converges to a small constant interval around $\OPT_{\beta}$ in a constant number of steps. 
        Let $\eta > 0$ such that $f_{\beta}'$ is bounded by some $L < 1$ on $[\OPT - \eta , \OPT + \eta]$.
        
        There exists a large constant $C$ and random variables $(\delta_t)_{0 \le t \le T}$ such that $|\delta_t| < \frac{C}{N^{1/2 - \eps}}$ almost surely for all $t < T$ and $\bz_{t+1} = f_{\beta}(\bz_t) + \delta_t$.
        We claim that after some constant number of steps $\tau$, $|z_{\tau}| \in [\OPT-\eta/2 , \OPT+\eta/2]$. Indeed, set
        \[ \mathrm{gap} = \inf_{ \substack{ h \in \left(\calH \cap [0,2]\right) \setminus \left[ \OPT - \eta/2 , \OPT + \eta/2 \right] } } \left| f_{\beta}(h) - h \right| \mcom \]
        a constant bounded away from $0$ due to the definition of $\calH$ and the boundedness of $f_{\beta}'$ near $\OPT$. Then, while $|\bz_t| \in \calH \setminus \left[ \OPT - \eta/2 , \OPT + \eta/2 \right]$, the fact that $f_\beta$ is increasing (\Cref{lem:fixed-point}(ii)), we have 
        \[ \left| |\bz_{t+1}| - \OPT \right| < \left| |\bz_t| - \OPT \right| - \mathrm{gap} + |\delta_t| \le \left| |\bz_t| - \OPT \right| - \frac{\mathrm{gap}}{2} \mper \]
        Furthermore, once $|\bz_t|$ reaches $[\OPT - \eta/2, \OPT + \eta/2]$, because $\delta_t = o(1)$, it is easy to check that $f_\beta'(z_t) < L$ for all future times $t > \tau$. Let us assume without loss of generality that $\bz_t$ is positive. For all such times, we have
		\[ \left| \bz_{t+1} - \OPT \right| = \left| f_\beta(\bz_t) + \delta_t - f_\beta(\OPT) \right| \le L \left| \bz_t - \OPT  \right| + \frac{C}{N^{1/2 - \eps}} \mper  \]
		Iterating this inequality yields that
	    \begin{equation}
	        \label{eq:one-dim-fixed-point-convergence}
	        \left| \bz_{\tau+t} - \OPT \right| \le L^{t} \left| \bz_\tau - \OPT \right| + \frac{1}{1-L} \cdot \frac{C}{N^{1/2 - \eps}}\mper
	    \end{equation}
		It follows that once we set $T = \Omega(\log N)$,
		\[ \left|\bz_T - \OPT\right| \le O\left(\frac{1}{N^{1/2 - \eps}}\right) \]
		as desired.
	\end{proof}

\subsection{High-temperature RGD escapes the unstable fixed point}
This section is dedicated to proving the following lemma about RGD escaping from the trivial fixed point at $0$ correlation.
\begin{lemma}
	\label{th:rgd-cold-start-analysis}
	Let $\lambda > 1$, $\beta \in (\frac{1}{\lambda}, 1)$, and $\eps > 0$ be a sufficiently small constant. Let $x_0 \in \{\pm 1\}^n$ arbitrary, and $(x_t)_{t \ge 0}$ the (random) trajectory of RGD initialized at $x_0$. Then,
	\[ \Pr\left[ \max_{0 \le t \le T} \frac{1}{N} \left|\left\langle x_t , \bone \right\rangle\right| \le \eps \right] \le e^{-\Omega\left( T / \log N \right)} \mper \]
\end{lemma}

\Cref{th:rgd-cold-start-analysis} almost immediately follows from the following more general lemma about appropriate stochastic processes escaping unstable fixed points.

\begin{lemma}
	\label{lem:stoch-proc-escape-unstable}
	Let $g : \R \to \R$ be a smooth function such that for some constants $L > 1$ and $\eps > 0$, $|g(x)| > L|x|$ for all $|x| < 2\eps$. Let $X_0 \sim \nu_0$ for some arbitrary distribution $\nu_0$ on $\R$. Suppose we have three sequences $(X_t,Y_t,Z_t)_{t \ge 0}$ of random variables on a common probability space.
    Assume that
	\begin{enumerate}[label=\normalfont{(\roman*)}]
		\item for all $t \ge 0$, $Z_t$ has law $\calN\left( 0 , \Theta\left( \frac{1}{N} \right) \right)$, and is independent of $\sigma( (X_s)_{0 \le s \le t-1} , (Y_s)_{0 \le s \le t} )$,
        \item $\E \left[ Y_t^2 \mid (X_s,Y_s,Z_s)_{0 \le s \le t-1} \right] \le \frac{C}{N}$.
        \item $(X_t)$ is defined by
        \[ X_{t+1} = g(X_t) + Y_t + Z_t \mper \]
	\end{enumerate}
	Then,
	\[ \Pr\left[ \max_{0 \le t \le T} |X_t| \le \eps \right] \le e^{-\Omega\left(T / \log N\right)} \mper \]
\end{lemma}

\begin{proof}
    Our proof proceeds in two stages. We first show that the dynamics escape zero to a ``lukewarm'' start, and after this exponential growth kicks in to make $X_t$ grow to $\Omega(1)$.

    To be more precise, let $\kappa > 0$ be a (large) constant to be fixed later. We shall show the following two statements, where $c > 0$ is a sufficiently small constant:
    \begin{align}
        \Pr\left[ \left| X_{t+1} \right| \ge \frac{\kappa}{\sqrt{N}} \mid (X_s,Y_s)_{0 \le s \le t}, (Z_s)_{0 \le s \le t-1} \right] &\ge c \mcom \label{eq:escape-0-one-step} \\
        \Pr\left[ \max_{t \le s \le t + C \log N} \left| X_{s} \right| \ge \eps \mid (X_s,Y_s)_{0 \le s \le t}, (Z_s)_{0 \le s \le t-1} , |X_{t}| \ge \frac{\kappa}{\sqrt{N}} \right] &\ge c \label{eq:escape-0-exp-growth} \mper
    \end{align}
    Putting these two together with the independence structure of the $(Y_s,Z_s)$ immediately yields the result.
    
    Let us start by proving \eqref{eq:escape-0-one-step}, which is a consequence of anticoncentration and concentration of $Z_t$. Indeed, we have
    \begin{align*}
        &\Pr\left[ \left| X_{t+1} \right| \ge \frac{\kappa}{\sqrt{N}} \mid \left(X_s,Y_s)_{0 \le s \le t} , (Z_s\right)_{0 \le s \le t-1} , \left| g(X_{t}) + Y_{t} \right| \ge \frac{2\kappa}{\sqrt{N}} \right] \\
        &\qquad\qquad\ge \Pr\left[ \left| Z_{t} \right| \le \frac{\kappa}{\sqrt{N}} \mid \left(X_s,Y_s)_{0 \le s \le t} , (Z_s\right)_{0 \le s \le t-1} , \left| g(X_{t}) + Y_{t} \right| \ge \frac{2\kappa}{\sqrt{N}} \right] \ge c \mcom \text{ and} \\
        &\Pr\left[ \left| X_{t+1} \right| \ge \frac{\kappa}{\sqrt{N}} \mid \left(X_s,Y_s)_{0 \le s \le t} , (Z_s\right)_{0 \le s \le t-1} , \left| g(X_{t}) + Y_{t} \right| \le \frac{2\kappa}{\sqrt{N}} \right] \\
        &\qquad\qquad\ge \Pr\left[ \left| Z_{t} \right| \ge \frac{3\kappa}{\sqrt{N}} \mid \left(X_s,Y_s)_{0 \le s \le t} , (Z_s\right)_{0 \le s \le t-1} , \left| g(X_{t}) + Y_{t} \right| \le \frac{2\kappa}{\sqrt{N}} \right] \ge c \mcom
    \end{align*}
    implying \eqref{eq:escape-0-one-step}.

    Let us next prove \eqref{eq:escape-0-exp-growth}. Consider the following event, parametrized by some constant $\iota > 0$ that we shall fix later:
    \[ \calE = \left\{ (Y_s,Z_s)_{s \ge t} : \left| Y_s + Z_s \right| \le \frac{\iota\kappa}{\sqrt{N}} \cdot (s-t+1) \right\} \mper \]
    Observe that if $\iota\kappa$ is sufficiently large, the above holds with nonzero probability. Indeed, we have $\E \left[ (Y_t + Z_t)^2 \mid (Y_s,Z_s)_{0 \le s \le t-1} \right] \le \frac{C}{N}$, so Markov's inequality followed by a union bound gives that
    \begin{equation}
        \label{eq:escape-0-union-bound-event}
        \Pr\left[ \calE \mid (X_s,Y_s,Z_s)_{0 \le s \le t-1} \right] \ge 1 - \sum_{j \ge 0} \frac{C}{\iota^2\kappa^2 (j+1)^2} \mcom
    \end{equation}
    which is bounded away from $0$ for $\iota\kappa$ sufficiently large. 
    
    Now, towards \eqref{eq:escape-0-exp-growth}, we claim the following for sufficiently small $\iota$ and some constant $\wt{L} \in \left( 1 , L \right)$ bounded away from $1$ (to be set later): if $\calE$ holds, $|X_{t+j}| \le \eps$, and $|X_{t+j}| \ge \wt{L}^j \cdot \frac{\kappa}{\sqrt{N}}$, then $|X_{t+j+1}| \ge \wt{L}^{j+1} \cdot \frac{\kappa}{\sqrt{N}}$. It is not difficult to see that this implies \eqref{eq:escape-0-exp-growth}: all the $X_s$ being smaller than $\eps$ would contradict the exponential growth prescribed by the previous sentence if the constant $C$ is taken large enough that $\wt{L}^{C \log N} \cdot \frac{\kappa}{\sqrt{N}} \ge \eps$. To complete the proof, we have that if the events described earlier hold,
    \begin{align*}
        \left| X_{t+j+1} \right| &= \left| g(X_{t+j}) + Y_{t+j} + Z_{t+j} \right| \\
            &\ge |g(X_{t+j})| - \left| Y_{t+j} + Z_{t+j} \right| \\
            &\ge L \cdot \wt{L}^j \cdot \frac{\kappa}{\sqrt{N}} - \frac{\iota\kappa}{\sqrt{N}} \cdot (j+1) \\
            &\ge \wt{L}^j \cdot \frac{\kappa}{\sqrt{N}} \cdot \left( L - \iota \cdot \sup_{j \ge 0} \frac{j+1}{\wt{L}^j} \right) \mper
    \end{align*}
    To conclude, choose the constants such that
    \begin{itemize}
        \item $\wt{L} = \frac{L+1}{2}$ is bounded away from $1$,
        \item $\iota$ is small enough that $L - \iota \cdot \sup_{j \ge 0} \frac{j+1}{\wt{L}^j} \ge \wt{L}$, and
        \item $\kappa$ is large enough that the probability in \eqref{eq:escape-0-union-bound-event} is bounded away from $0$. \qedhere
    \end{itemize}
\end{proof}

Let us go back and prove \Cref{th:rgd-cold-start-analysis}.

\begin{proof}[Proof of \Cref{th:rgd-cold-start-analysis}]
	We shall prove a similar statement not for the correlation, but for the external fields $(\bh_t)_{t \ge 0}$ encountered along the trajectory of RGD (see \Cref{eq:gtilde}), normalized by a factor of $\frac{1}{\beta\lambda}$. Concretely, this process is defined as follows, starting from arbitrary $h = \bh_0 \in \R$.
    \begin{enumerate}
    	\item sample $\bz = \boldm_{\beta\lambda \bh_t} + \frac{1}{\sqrt{N}} \wt{\bg}_{\beta\lambda\bh_t}$, where we recall $\wt{\bg}_h$ is the deviation from the expected magnetization of a Gibbs sample from $\mu_{\beta \bW, h}$.
        \item draw $\bg \sim \mathcal{N}\left(0, \frac{1}{\beta\lambda N}\right)$, and move to $\bh' = \bz + \bg$.
    \end{enumerate}
    \Cref{lem:all-field-mag-conc} yields that with probability $1 - o(1)$, we may write 
	\[ \bh_{t+1} = f_{\beta,\lambda}(\bh_{t}) \left( 1 + O\left( \frac{1}{N^\eps} \right) \right) + O\left( \frac{1}{N^{1/2 + \eps}} \right) + \frac{1}{N^{1/2}} \cdot \wt{\bg}_{\beta\lambda\bh_t}+ \bg_t \mper \]
    Let us verify that we can apply \Cref{lem:stoch-proc-escape-unstable} with $X_t = \bh_t$, $Y_t = O\left( \frac{1}{N^{1/2 + \eps}} \right) + \frac{1}{N^{1/2}} \cdot \wt{\bg}_{\beta\lambda \bh_t}$, and $Z_t = \bg_t$.
    
    By definition, $\bg_t$ is independent of the collection $\qty{\wt{\bg}_{\beta\lambda\bh_t},  (\wt{\bg}_{\beta\lambda\bh_s}, \bg_s)_{0 \le s \le t-1}}$. The bound on the conditional second moment of $\wt{\bg}_h$ follows from \Cref{lem:strong-field-magnetization}---recall that for $\beta < 1$, $(\beta,h)$ satisfies \eqref{eq:overlap-conc} for all $h$. The lower bound of $|f_{\beta,\lambda}(x)| \left( 1 + O\left( \frac{1}{N^{\eps}} \right) \right) > L |x|$ for small $x$ follows from the instability of $f_{\beta,\lambda}$ at $0$ for $\beta > \frac{1}{\lambda}$: this is \Cref{lem:fixed-pt-structure-unconditional}(ii).
    The desideratum follows on using \Cref{lem:stoch-proc-escape-unstable}.
 \end{proof}


We may now prove \Cref{th:rgd-fixed-temp}, restated for convenience.

\thrgdfixedtemp*

\begin{proof}
    If $\beta \ge 1$, this immediately follows from \Cref{lem:rgd-convergence-from-warm-start}, since it is then non-trivial to satisfy the AT line with a margin. For $\beta < 1$, this follows from \Cref{th:rgd-cold-start-analysis,lem:rgd-convergence-from-warm-start}.
\end{proof}


\section{High-temperature mean magnetization estimates in the SK model}

\label{subsec:mean-magnetization}

In the RGD recursion for spiked Wigner, the measure decomposition is \[
    \mu_{\beta M, 0} = \E_{\substack{x \sim \mu_{\beta M, 0} \\ \bg \sim N(0,1)}}[\mu_{\beta W, (\frac{\beta \lambda}{n} \ev{\bone, x} + \sqrt{\beta \lambda /n}\bg )\bone}].
\]
The following explicit formula allows us to determine the deterministic scalar recursion which determines the behavior of locally stationary RGD, in the high-temperature regime $\beta < 1$.

\begin{restatable}[High-temperature magnetization concentration for all fields]{lemma}{skformula}
    \label{lem:all-field-mag-conc}
    Let $\beta < 1$, $\eps > 0$ sufficiently small, and $K > 0$ large independent of $N$. With probability $1-o(1)$ over $W \sim \GOE(N)$, for all $0 \le h \le K$,
    \begin{enumerate}[label=(\roman*)]
        \item if $h < N^{- \left( 1/4 + \eps \right)}$,
        \[ \frac{1}{N} \E_{\bx \sim \mu_{\beta W,h\bone}} \langle\bx,\bone\rangle = \E_{\bg \sim \calN(0,1)}\left[ \tanh\left( \beta\bg\sqrt{q} + h \right) \right] \left( 1 + O\left(\frac{1}{N^{\eps}}\right) \right) + O\left( \frac{1}{N^{1/2 + \eps}} \right) \mper \]

        \item if $h > N^{- \left( 1/4 + \eps \right)}$,
        \[ \frac{1}{N} \E_{\bx \sim \mu_{\beta W,h\bone}} \langle\bx,\bone\rangle = \E_{\bg \sim \calN(0,1)}\left[ \tanh\left( \beta\bg\sqrt{q} + h \right) \right] + O \left( \frac{1}{N^{1/2 - \eps}} \right) \mper \]

    \end{enumerate}
    where $q = q_h$ is the functional order parameter defined by
    \[ q = \E_{\bg \sim \calN(0,1)} \left[\tanh^2 \left( \beta\bg\sqrt{q} + h \right)\right]. \]
\end{restatable}

Note that for $h > N^{-(1/4 + \eps)}$, a bound of the form of (i) is strictly stronger than that in (ii).



The above theorem will be a corollary of two estimates for the mean magnetization. The first is essentially already present in \cite{Han07} (also see \cite{CT21}).

\begin{restatable}[{essentially \cite{Han07}}]{lemma}{strongfieldmeanmagnetization}
    \label{lem:strong-field-magnetization}
    Let $(\beta,h)$ satisfy \eqref{eq:overlap-conc}. Then, for any constant $k$,
    \[ \E \left[\left( R(\sigma,\bone) - q_1 \right)^{2k} \right] \le O\left( \frac{1}{N^k} \right) \mcom \]
    where we denote $q_1 = \E_{\bg \sim \calN\left(0,1\right)} \left[ \tanh\left( \beta \bg \sqrt{q} + h \right) \right]$.
\end{restatable}

The proof is identical to that of Corollary 4.1, Theorem 5.1, and Theorem 6.1 in \cite{Han07}---we observe that all these proofs of upper bounds on the moments go through assuming overlap concentration, even if we do not have a central limit theorem for the overlap (such a CLT \emph{is} required for the stronger Theorem 1.2 in Hanen's paper, which proves a CLT for the magnetization).

However, this lemma is not useful in the regime where $h$ is tiny, say $\Theta\left(N^{-1/2}\right)$, since $q_1$ itself is $O(h)$, and the error is at the same scale as the estimate. To get around this, we prove a mean magnetization formula that is more precise in the small $h$ regime, albeit only with a lower moment that will suffice for our purposes.

\begin{restatable}[Magnetization estimates under weak external field]{lemma}{weakfieldmagnetization}
    \label{lem:weak-field-magnetization}
    Let $\delta > 0$ and $h < N^{-\alpha}$ for some $\alpha > 1/4$. Then,
    \[ \E \left[ \left| \frac{1}{N} \sum_{i=1}^{N} \langle \sigma_i\rangle - q_1 \right|^{2-\delta}  \right] \le O\left( \max\left\{ \frac{1}{N} , \frac{h}{\sqrt{N}}\right\} \right)^{2-\delta} \mper \]
\end{restatable}

We will defer the proof of \Cref{lem:weak-field-magnetization} to \Cref{subsec:weak-ext-field-magnetization}. We may now prove \Cref{lem:all-field-mag-conc}.

\begin{proof}
    For ease of notation, let $\boldm(h)$ be the random variable $\frac{1}{N} \sum_{i=1}^{N} \langle \sigma_i\rangle$ under external field $h$. It is not difficult to show that $\boldm(h)$ is non-decreasing in $h$: indeed, its derivative with respect to $h$ is the variance of the magnetization in the corresponding model. For $h > N^{- \left( \frac{1}{4} + \eps \right)}$, Markov's inequality applied to \Cref{lem:strong-field-magnetization} yields that
    \begin{equation}
        \label{eq:strong-field-high-prob-markov-estimate}
        \Pr\left[ \left| \boldm(h) - q_1(h) \right| > \frac{1}{N^{\frac{1}{2} - \eps}} \right] < \frac{1}{N^2} \mper
    \end{equation}
    On the other hand, for $h < N^{- \left( \frac{1}{4} + \eps \right)}$, Markov's inequality applied to \Cref{lem:weak-field-magnetization} yields that for some $c_1$ we will fix shortly,
    \[ \Pr\left[ \left| \boldm(h) - q_1(h) \right| > \max\left\{ \frac{h}{N^{ \frac{1}{2} - \frac{c_1}{2-\delta} }} , \frac{1}{N^{1 - \frac{c_1}{2-\delta}}} \right\} \right] < O\left(N^{-c_1}\right) \mper \]
    Now, choose $\delta = \eps$ and $c_1$ such that $\frac{c_1}{2-\eps} = \frac{1}{2}-\eps$, so appealing to \Cref{lem:q1-q-bound}, the above reduces to
    \begin{equation}
        \label{eq:weak-field-high-prob-markov-estimate}
        \Pr\left[ \left| \boldm(h) - q_1(h) \right| > \max\left\{ \frac{q_1(h)}{N^{\eps}} , \frac{1}{N^{1/2 + \eps}} \right\} \right] < O\left( N^{-c_1} \right) \mper
    \end{equation}
    
    For $c_2 = \frac{1}{2} + \eps$, let $S = \left\{ q_1^{-1}\left( i N^{-c_2} \right) : 0 \le i \le KN^{c_2} \right\}$; note that $S$ is well defined because $q_1$ is strictly increasing in $h$. We shall perform a union bound over $S$. For sufficiently small $\eps$, it is plainly true that $N^{-c_1} \cdot N^{c_2} = o(1)$, since $c_1 \approx 1$ and $c_2 \approx \frac{1}{2}$. Consequently, \eqref{eq:weak-field-high-prob-markov-estimate} yields that with probability $1-o(1)$, for all $h \in S$ with $h < N^{- \left( \frac{1}{4} + \eps \right)}$,
    \[ \left| \boldm(h) - q_1(h) \right| < O\left(\max\left\{ \frac{q_1(h)}{N^{\eps}} , \frac{1}{N^{1/2 + \eps}} \right\}\right) \mper \]
    On the other hand, \eqref{eq:strong-field-high-prob-markov-estimate} yields that with probability $1-o(1)$, for all $h \in S$ with $h > N^{- \left( \frac{1}{4} + \eps \right)}$,
    \[ \left| \boldm(h) - q_1(h) \right| < \frac{1}{N^{1/2 - \eps}} \mper \]
    This bound further extends from $S$ to all $h$ bounded by $K$: indeed, suppose $h \le K$ and $h_1 < h < h_2$ such that $h_1,h_2 \in S$, $|q_1(h_1) - q_1(h)| < N^{-c_2}$ and $|q_1(h_2) - q_1(h)| < N^{-c_2}$. Then, if $h < N^{- \left( \frac{1}{4} + \eps \right)}$,
    \[ \boldm(h) \ge \boldm(h_1) \ge q_1\left(h_1\right) - O\left(\frac{q_1(h_1)}{N^{\eps}}\right) - O\left(\frac{1}{N^{1/2 + \eps}}\right) \ge q_1(h) - O\left(\frac{q_1(h)}{N^{\eps}}\right) - O\left(\frac{1}{N^{1/2 + \eps}}\right) \mcom \]
    where we have used the fact that $q_1$ is strictly increasing in $h$.
    Similarly, we have 
    \[ \boldm(h) \le q_1(h) + O\left(\frac{q_1(h)}{N^{\eps}}\right) + O\left(\frac{1}{N^{1/2 + \eps}}\right) \mper \]
    An identical argument works for the alternate error bound when $h > N^{ - \left( \frac{1}{4} + \eps \right) }$.
\end{proof}

\subsection{High-precision mean magnetization estimates under weak external field}

\label{subsec:weak-ext-field-magnetization}

For this section, let $h < N^{-\alpha}$ for some $\alpha > 1/4$, and $\beta < 1$. We also define $q,q_1$ by the recursions
\begin{align*}
    q &= \E_{\bz \sim \calN(0,1)} \left[ \tanh^2 \left( \beta \sqrt{q} \bz + h \right)  \right] \mcom \\
    q_1 &= \E_{\bz \sim \calN(0,1)} \left[ \tanh\left( \beta\sqrt{q} \bz + h \right)  \right] \mper
\end{align*}
We also denote $Z$ to be the partition function of the SK model with external field $h$.

In this section, we will prove the following lemma.

\weakfieldmagnetization*

\begin{lemma}\label{lem:mag-second-moment}
    $\displaystyle \E \left[ Z^2 \left( \frac{1}{N} \sum_{i=1}^{N} \langle \sigma_i\rangle - q_1 \right)^{2}  \right] \le \E[Z]^2 \cdot O\left( \max\left\{ \frac{1}{N} , \frac{h}{\sqrt{N}}\right\}  \right)^2 \mper$
\end{lemma}

\begin{proof}[Proof of \Cref{lem:weak-field-magnetization}]
    We have by H\"{o}lder's inequality and \cite[Corollary 3.5]{DW23} (using the fact that $\alpha > 1/4$) that
    \begin{align*}
        \E \left[  \left| \frac{1}{N} \sum_{i=1}^{N} \langle \sigma_i\rangle - q_1\right|^{2-\delta} \right] &\le \E\left[ Z^2 \cdot \left( \frac{1}{N} \sum_{i=1}^{N} \langle \sigma_i\rangle - q_1 \right)^2  \right]^{\frac{2-\delta}{2}} \cdot \E\left[ \frac{1}{Z^{4/\delta - 2}} \right]^{\delta/2} \\
            &\le \E[Z]^{2-\delta} \cdot O\left( \max\left\{ \frac{1}{N} , \frac{h}{\sqrt{N}}\right\} \right)^{2-\delta} \cdot \frac{1}{\E[Z]^{2-\delta}} \\
            &= O\left( \max \left\{ \frac{1}{N} , \frac{h}{\sqrt{N}}\right\} \right)^{2-\delta} \mcom
    \end{align*}
    as desired.
\end{proof}

We dedicate the remainder of this section to proving \Cref{lem:mag-second-moment}.
Using replicas, we can write 
\[
\E\left[Z^2 \left( \frac{1}{N} \sum_{i=1}^{N} \langle \sigma_i\rangle - q_1 \right)^{2}  \right]
= \sum_{\sigma, \rho} \E[e^{H_N(\sigma) + H_N(\rho)} (m(\sigma) - q_1)(m(\rho) - q_1)],\]
where $H_N(\sigma) = \frac{\beta}{2} \langle\sigma, \bW \sigma\rangle + h \ev{\sigma, 1}$, so that $\Cov(H_N(\sigma), H_N(\rho)) = \frac{\beta^2}{2} NR(\sigma, \rho)^2$. It follows that $H_N(\sigma) + H_N(\rho)$ is distributed as $\calN\left( h \left\langle \sigma+\rho , \bone\right\rangle , N\beta^2 \left( 1 + R(\sigma,\rho)^2 \right) \right)$, and so,
\[ \E \left[ e^{H_N(\sigma) + H_N(\rho)}  \right] = e^{ \frac{N\beta^2}{2} \left( 1 + R(\sigma,\rho)^2 \right) + h \left\langle \sigma+\rho , \bone\right\rangle  } \mper \]

Since the summand in fact has no dependence on the disorder, we have reduced the moment calculation to an evaluation of the double sum 
\begin{align*}
    &\sum_{\sigma, \rho \in \{\pm 1\}^N} e^{\frac{N\beta^2}{2}(1 + R(\sigma, \rho)^2) + h\ev{\sigma + \rho, \bone}} (m(\sigma) - q_1)(m(\rho) - q_1) \\
    &\qquad \qquad = (\E Z)^2 (2 \cosh(h))^{-2N} \sum_{\sigma, \rho} e^{\frac{N\beta^2}{2} R(\sigma, \rho)^2 + h\ev{\sigma + \rho, \bone}}(m(\sigma) - q_1)(m(\rho) - q_1)\mcom
\end{align*}
where we have used that $\E[Z] = (2\cosh(h))^Ne^{\frac{\beta^2N}{4}} $. This explains the dependence on $(\E Z)^2$.

It thus remains to compute the double sum. To this end, we will use the HS transform, followed by the Laplace method. 
\begin{fact}[Hubbard--Stratonovich transform]\label{fact:hs}
    For any $c > 0$, we have 
    \[
        e^{cx^2} = \frac{1}{\sqrt{c\pi}}\int_{\R} e^{-\frac{z^2}{c} + 2zx}\dif z
    \]
\end{fact}
Applying the HS transform with $c = \frac{2\beta^2}{N}$ and $x = \frac{\ev{\sigma, \rho}}{2}$, we obtain that 
\[
    e^{\frac{N\beta^2}{2} R(\sigma, \rho)^2} = e^{\frac{\beta^2}{2N} \ev{\sigma, \rho}^2} = \sqrt{\frac{N}{2\pi\beta^2}} \int_{\R} e^{-\frac{N}{2\beta^2} z^2 + z\ev{\sigma, \rho}} \dif z,
\]
so that the double sum becomes 
\begin{align*}
    \sqrt{\frac{N}{2\pi\beta^2}} \cdot \int_{\R} \exp(-\frac{N}{2\beta^2} z^2) \sum_{\sigma, \rho} \exp(z\ev{\sigma, \rho} + h\ev{\sigma + \rho, \bone}) (m(\sigma) - q_1)(m(\rho) - q_1) \dif z.
\end{align*}
Observe that given $z$, the pairs $(\sigma_i, \rho_i)$ are independent of each other, and hence this sum over $\sigma, \rho$ is tractable to explicitly compute. 
In particular, consider the distribution $p$ on $\{\pm 1\}^2$, with $p(a, b) = \frac{1}{C(z, h)} \exp(zab + ha + hb)$. Then one can check the normalization constant is $C(z, h) = 2(e^z \cosh(2h) + e^{-z})$, and we have
\[  \sum_{\sigma, \rho} \exp(z\ev{\sigma, \rho} + h\ev{\sigma + \rho, \bone}) (m(\sigma) - q_1)(m(\rho) - q_1) = C(z,h)^N \cdot \E_{p^{\otimes N}} \left[ \left(m(\sigma) - q_1\right) \left( m(\rho) - q_1 \right)  \right] \mper \]

Now, expanding it out, we have
\[ (m(\sigma) - q_1)(m(\rho) - q_1) = \frac{1}{N^2} \sum_{i \neq j} (\sigma_i - q_1)(\rho_j - q_1) + \frac{1}{N^2} \sum_{i} (\sigma_i - q_1)(\rho_i - q_1) \mcom \]
and therefore, 
\begin{align*}
    f_N(z) &\defeq \E_{p^{\otimes N}}[(m(\sigma) - q_1)(m(\rho) - q_1)] \\
        &= \frac{N(N-1)}{N^2} (m(z) - q_1)^2 + \frac{1}{N} (k(z) - 2q_1 m(z) + q_1^2) \\
        &= \left( m(z) - q_1  \right)^2 + \frac{1}{N} \left( k(z) - m(z)^2 \right) \\
        &= \left( m(z) - q_1 \right)^2 + \frac{1}{N} \left( k(z) - q \right) + \frac{1}{N} \left( q - q_1^2 \right) + \frac{1}{N} \left( q_1^2 - m(z)^2 \right)
\end{align*}
where $m(z) \defeq \E_{p}[a] = \E_{p}[b]$ and $k(z) \defeq \E_{p}[ab]$. 
Thus, our goal is to show that
\begin{align*}
    \left( 2 \cosh(h) \right)^{-2N} \cdot \sqrt{\frac{N}{2\pi\beta^2}} \cdot \int_{\R} \exp(N\qty[-\frac{z^2}{2\beta^2} + \log(2(e^z \cosh(2h) + e^{-z}))]) f_N(z) \dif z \le O\left( \frac{h^2}{N} \right) \mper
\end{align*}
Rearranging, simplifying, and summarizing, we would like to show that
\[ \E_{\bz \sim \calN\left( 0 , \frac{\beta^2}{N} \right)} \left[ \left( \frac{ e^z \cosh(2h) + e^{-z} }{ \cosh(2h) + 1  } \right)^N f_N(z) \right] \le O\left( \frac{h^2}{N} \right) \mcom \]
where an explicit calculation yields that
\begin{align*}
    m(z) &= \frac{e^z \sinh(2h)}{e^z \cosh(2h) + e^{-z}} \text{ and} \\
    k(z) &= \frac{e^z \cosh(2h) - e^{-z}}{e^z \cosh(2h) + e^{-z}}\mper
\end{align*}
Let us massage this expression to deal with the first part of the integrand. Set
\[ g_N(z) = \frac{\cosh(z) + \sinh(h)^2 e^z}{ \left( 1 + \sinh(h)^2 \right) e^{\tanh(h)^2 \cdot z} \left( e^{\frac{z^2}{2} \left( 1 - \tanh(h)^4 \right) }  \right)  } \mper \]
Then,
\[ (*) = \E_{\bz \sim \calN\left( 0 , \frac{\beta^2}{N} \right)} \left[ g_N(\bz)^N \cdot e^{ N \tanh(h)^2 \bz } \cdot e^{ \frac{N\bz^2}{2} \left( 1 - \tanh(h)^4 \right) } \cdot f_N(\bz) \right] \]
Define $\gamma > 0$ by $\frac{1}{\gamma^2} = \frac{1}{\beta^2} - \left( 1 - \tanh(h)^4 \right)$, and $\mu = \gamma^2 \tanh(h)^2$, so
\begin{align*}
    (*) &= \frac{\gamma}{\beta} \cdot \E_{\bz \sim \calN\left( 0 , \frac{\gamma^2}{N} \right)} \left[ g_N(\bz)^N \cdot e^{ N \tanh(h)^2 \bz } \cdot f_N(\bz) \right] \\
        &= \frac{\gamma}{\beta} \cdot e^{ \frac{1}{2} \cdot N \gamma^2 \tanh(h)^4 } \cdot \E_{\bz \sim \calN\left( \mu , \frac{\gamma^2}{N} \right)} \left[ g_N(\bz)^N \cdot f_N(\bz) \right] \\
        &\asymp \E_{\bz \sim \calN\left( \mu , \frac{\gamma^2}{N} \right)} \left[ g_N(\bz)^N \cdot f_N(\bz) \right] \mcom
\end{align*}
where the final line uses the fact that $h \ll N^{-1/4}$. Why is this useful? By the choice of $g$, we in fact have $g(0) = 1$, $g^{(1)}(0) = g^{(2)}(0) = 0$, and $g^{(3)}(0) = -2 \tanh(h)^2 \left( 1 - \tanh(h)^4 \right) = O(h^2)$. Thus, since $\bz$ is typically $\lesssim \frac{1}{\sqrt{N}}$, $g_N(\bz)$ is typically $1 + O\left( \frac{h^2}{N^{3/2}} \right)$, and so $g_N(\bz)^N$ is typically $1+O\left( \frac{h^2}{\sqrt{N}} \right)$, essentially allowing us to ignore it in the integral.

More concretely, for our purposes (where we do not care about precise constants), it suffices to use the looser Cauchy--Schwarz bound, yielding
\begin{align*}
    (*) \lesssim \E \left[ g_N(\bz)^{2N} \right]^{1/2} \cdot \E \left[ f_N(\bz)^2 \right]^{1/2}
\end{align*}
To control this, we will show
\begin{enumerate}[label=(\roman*)]
    \item $ \E \left[ \left( m(\bz) - q_1 \right)^{2k} \right] \le O\left( \frac{h^{2k}}{N^{k}} \right)$ for all constant $k$, proved in \Cref{prop:magnetization-control},
    \item $ \E \left[ \left( k(\bz) - q \right)^{2k} \right] \le O\left( \frac{1}{N^{k}} \right)$ for all constant $k$, proved in \Cref{prop:overlap-control}, and
    \item $\E \left[ g_N(\bz)^{2N} \right] = O(1)$, proved in \Cref{prop:reweigh-control}.
\end{enumerate}
These essentially immediately yield \Cref{lem:mag-second-moment}.

\subsection{A series of Gaussian expectations}

\begin{proposition}[Controlling the magnetization]
    \label{prop:magnetization-control}
    For any constant $k$,
    \[ \E_{\bz \sim \calN\left( \mu , \frac{\gamma^2}{N} \right)} \left[ \left( \frac{ e^{\bz} \sinh(2h) }{ e^{\bz} \cosh(2h) + e^{-z} } - q_1 \right)^{2k} \right] \le O\left( \frac{h^2}{N} \right)^k \mper \]
\end{proposition}
\begin{proof}
    Let us write
    \begin{align*}
       &\E \left[ \left( \frac{ e^{\bz} \sinh(2h) }{ e^{\bz} \cosh(2h) + e^{-z} } - q_1 \right)^{2k} \right] \\
       &\qquad\lesssim \E \left[ \left( \frac{ e^{\bz} \sinh(2h) }{ e^{\bz} \cosh(2h) + e^{-z} } - \E \left[ \frac{ e^{\bz} \sinh(2h) }{ e^{\bz} \cosh(2h) + e^{-z} } \right] \right)^{2k} \right] + \left( \E \left[\frac{ e^{\bz} \sinh(2h) }{ e^{\bz} \cosh(2h) + e^{-z} }\right] - q_1 \right)^{2k} \\
       &\qquad= O\left( h^{2k} \right) \E \left[ \left( \frac{ e^{\bz} }{ e^{\bz} \cosh(2h) + e^{-z} } - \E \left[ \frac{ e^{\bz} }{ e^{\bz} \cosh(2h) + e^{-z} } \right] \right)^{2k} \right] + \left( \E \left[\frac{ e^{\bz} \sinh(2h) }{ e^{\bz} \cosh(2h) + e^{-z} }\right] - q_1 \right)^{2k} \mper
    \end{align*}
    The first term is easy to deal with: because $z \mapsto \frac{ e^{\bz} }{ e^{\bz} \cosh(2h) + e^{-z} }$ is an $O(1)$-Lipschitz function, Gaussian concentration of Lipschitz functions yields that the associated $2k$th moment is bounded by $O\left( \frac{\gamma^2}{N} \right)^{k}$.

    For the second term, because $q = O(h^2)$,
    \[ q_1 = \E_{\bg \sim \calN\left(0,1\right)} \left[ \tanh\left( \beta \bg \sqrt{q} + h \right) \right] = h + O(h^3). \]
    On the other hand, setting $f(z) = \frac{ e^{\bz} \sinh(2h) }{ e^{\bz} \cosh(2h) + e^{-z} }$, so
    \begin{align*}
        f(0) &= \tanh(h) = h + O(h^3)\mcom \\
        f'(0) &= O(h)\mcom \text{ and } \\
        f''(0) &= O\left(h^3\right)\mcom
    \end{align*}
    we have
    \begin{align*}
        \E \left[ f(\bz) \right] &= f(0) + f'(0) \mu + \frac{1}{2} f''(0) \left( \mu^2 + \frac{\gamma^2}{N} \right) + O\left( \frac{h}{\sqrt{N}} \right) \\
            &= h + O(h) \cdot \gamma^2 \tanh(h)^2 - O(h^3) \cdot \left( \mu^2 + \frac{\gamma^2}{N} \right) + O\left(\frac{h}{\sqrt{N}}\right) \\
            &= h + O\left(\frac{h}{\sqrt{N}}\right).
    \end{align*}
    Consequently,
    \[ \left|  \E \left[ \frac{ e^{\bz} \sinh(2h) }{ e^{\bz} \cosh(2h) + e^{-z} } \right] - q_1 \right| = O\left( \frac{h}{\sqrt{N}} \right) \mcom \]
    completing the proof.
\end{proof}

\begin{proposition}[Controlling the overlap]
    \label{prop:overlap-control}
    For any constant $k$,
    \[ \E_{ \bz \sim \calN\left( \mu , \frac{\gamma^2}{N} \right)  } \left[ \left( \frac{e^{\bz} \cosh(2h) - e^{-\bz}}{e^{\bz} \cosh(2h) + e^{-\bz}} - q \right)^{2k} \right] \le O\left(\frac{1}{N^k}\right) \mper \]
\end{proposition}
\begin{proof}
    Our strategy will essentially be identical to that of the previous proof. Again, because $z \mapsto k(z)$ is an $O(1)$-Lipschitz function, we have
    \[ \E \left[ \left( \frac{e^{\bz} \cosh(2h) - e^{-\bz}}{e^{\bz} \cosh(2h) + e^{-\bz}} - q \right)^{2k} \right] \lesssim O\left( \frac{1}{N^k} \right) + \left( \E \left[ \frac{e^{\bz} \cosh(2h) - e^{-\bz}}{e^{\bz} \cosh(2h) + e^{-\bz}} \right] - q \right)^{2k} \]
    We have $k(z) = \tanh\left( z + \frac{1}{2} \ln \cosh(2h) \right) = \tanh\left( z + h^2 + O(h^4) \right)$. As a result,
    \begin{align*}
       \E \left[ k(\bz) \right] &= \E\left[ \bz + h^2 - \frac{1}{3} \bz^3 \right] + O\left( h^4 \right) \\
        &= \mu + h^2 - \frac{1}{3} \cdot \mu^3 - \mu \cdot \frac{\gamma^2}{N} + O\left( h^4 \right) \\
        &= \mu + h^2 + O\left( \frac{1}{N} \right) \mper
    \end{align*}
    Let us also figure out the leading terms in the Taylor series of $q$. We have $q = O(h^2)$, so suppose that $q = a h^2 + O(h^4)$. Expanding out the recursion defining $q$, using the Taylor series expansion $\tanh^2(x) = x^2 + O(x^4)$, we have
    \[ q = \E_{\bg \sim \calN\left(0,1\right)}\left[ \beta\sqrt{q} \bg + h \right] = \beta^2 q + h^2 + O\left(h^4\right). \]
    It follows that $q = \frac{1}{1-\beta^2} \cdot h^2 + O(h^4)$. We also have
    \begin{align*}
        \mu + h^2 &= \tanh(h)^2 \cdot \frac{1}{ \frac{1}{\beta^2} - \left( 1 - \tanh(h)^4 \right) } + h^2 \\
            &= h^2 \left( \frac{1}{ \frac{1}{\beta^2} - 1 } + 1 \right) + O\left( \frac{1}{N} \right) \\
            &= \frac{1}{1-\beta^2} \cdot h^2 + O\left( \frac{1}{N} \right) \mper
    \end{align*}
    Therefore,
    \[ \left| \E\left[ k(\bz) \right] - q \right| = O\left( \frac{1}{N} \right), \]
    completing the proof.
\end{proof}

\begin{proposition}
    \label{prop:reweigh-control}
    It holds that $\E_{\bz \sim \calN\left( 0 , \frac{\beta^2}{N} \right)} \left[ g_N(\bz)^{2N} \right] = 1 + o_N(1)$.
\end{proposition}
\begin{proof}
    Recall
    \[ g_N(z) = \frac{\cosh(z) + \sinh(h)^2 e^z}{ \left( 1 + \sinh(h)^2 \right) e^{\tanh(h)^2 \cdot z} \left( e^{\frac{z^2}{2} \left( 1 - \tanh(h)^4 \right) }  \right)  } \mper \]
    Let $\kappa = \tanh(h)^2 = \frac{\sinh(h)^2}{1+\sinh(h)^2}$, so we have
    \[ g_N(z) = \frac{ \cosh(z) + \kappa \sinh(z) }{ e^{\kappa z + \frac{z^2}{2} \cdot (1-\kappa^2) } } \mper \]
    First note that the terms in the Taylor expansion up to order $2$ cancel out---indeed, for small $z$, we have $\cosh(z) + \kappa \sinh(z) = 1 + \kappa z + \frac{z^2}{2} + O(z^3)$, and
    \[ e^{\kappa z + \frac{z^2}{2} \cdot (1-\kappa^2)} = 1 + \kappa z + (1-\kappa^2) \cdot \frac{z^2}{2} + \frac{1}{2} (\kappa z)^2 + O(z^3) = 1 + \kappa z + \frac{z^2}{2} + O(z^3) \mper \]
    The claim near-immediately follows. For large $z$, the $e^{(1-\kappa^2)z^2/2}$ term in the denominator dominates, so we may establish a uniform bound of $g_N(z) \le 1 + O(|z|^3)$. We may thus bound $\E \left[ g_N(\bz)^{2N} \right] \le \E \left[ \left( 1 + O(|\bz|^3) \right)^{2N} \right] \le 1 + O \left( \frac{1}{\sqrt{N}} \right)$, as desired.
\end{proof}


\section{Understanding the RGD recursion}
\label{sec:understanding-fixed-points}


In this section, we analyze some salient properties of the fixed point equations governing RGD and AMP. In particular, we show that there exists a unique nonzero fixed point to $f_{\beta,\lambda}$ (as defined in \eqref{eq:f-beta-lambda}), and furthermore, this fixed point lies in the \eqref{eq:at-line} region and is stable.

\lemfixedptstructure*

We will prove the above over the course of this section.
We start by noting the following simple fact.

\begin{fact}\label{fact:sech-tanh}
    For any $a,b > 0$,
        \begin{align*}
            2a\E[\sech^2(a\bg + b)\tanh(a \bg + b)] = -\E_{\bg \sim \calN\left( 0,1 \right)} [\sech^2(a\bg + b) \cdot \bg] > 0.
        \end{align*}
\end{fact}
\begin{proof}
    The first equality is just Stein's lemma.
    For the second, we use symmetry to express
    \begin{align*}
        \E_{\bg \sim \calN\left( 0,1 \right)} [\sech^2(a\bg + b) \cdot \bg] &= \frac{1}{2} \cdot \E \left[\bg \left( \sech^2\left( b + a\bg \right) - \sech^2\left( b - a\bg \right) \right)\right] < 0 \mper
    \end{align*}
    If $\bg > 0$ we have $|b+a\bg| > |b-a\bg|$, and if $\bg < 0$, we have $|b-a\bg| > |b+a\bg|$. Since $\sech^2(\cdot)$ is even and decreasing for positive arguments, the claim follows.
\end{proof}

We also need the following facts about the recursion defining the overlap constant $q$.

\begin{lemma}\label{lem:q-properties}
    Let $\beta > 0$ and $h \ge 0$. Consider the fixed point equation 
    \[
    q = q(h) = \E[\tanh^2(\beta \bg \sqrt{q} + h)]\mper
    \]
    Then the following hold.
    \begin{enumerate}[label=\normalfont{(\roman*)}]
        \item If $h > 0$, there is a unique solution $q = q_{\beta,h} > 0$. If $h = 0$ and $\beta \le 1$, $q = q_{\beta,0} = 0$ is the unique solution.
        \item If $h = 0$ and $\beta > 1$, there is additionally a unique positive solution $q_{\beta,0} = \ol{q}_{\beta}$. Moreover,
        \[
            \lim_{h \downarrow 0} q_{\beta,h} =
            \begin{cases}
                0 \mcom & \beta \le 1 \mcom \\
                \ol{q}_{\beta} \mcom & \beta > 1 \mper
            \end{cases}
        \]
        \item If $(\beta, h)$ satisfy \eqref{eq:at-line}, then $q$ is differentiable. Denoting $\bX = \beta\bg\sqrt{q} + h$, we have
        \[
        q'(h) = \frac{2\E[\tanh(\bX)\sech^2(\bX)]}{1 - \beta^2 \E[\sech^4(\bX) - 2\sech^2(\bX)\tanh^2(\bX)]} \mper
        \]
        Consequently, $q'(h)> 0$.
    \end{enumerate}
\end{lemma}
\begin{proof}
    The first part of (i) is \cite{Gue01} (also see \cite[Proposition 1.3.8]{Tal10}). Now, for $q > 0$, the function $q \mapsto \frac{\E \tanh^2(\beta\bg\sqrt{q})}{q}$ is strictly decreasing. Indeed, $\frac{\E \tanh^2(\beta\bg\sqrt{q})}{\E (\beta\bg\sqrt{q})^2}$ is strictly decreasing because the function $t \mapsto \frac{\tanh(t)}{t}$ is strictly decreasing on $(0,\infty)$. Furthermore, it is not difficult to verify that
    \[ \lim_{q \to 0} \frac{\tanh^2(\beta\bg\sqrt{q})}{q} = \beta^2 \mper \]
    The second part of (i) and the first part of (ii) follow, as does the limiting behavior when $\beta \le 1$. For the statement about the limit, suppose instead that there were some sequence $(h_i) \downarrow 0$ such that $q_{\beta,h_i} \to 0$. Let $\bX_{i} = \beta\bg\sqrt{q_{\beta,h_i}} + h_i$. Because $0 \le t^2 - \tanh^2 t \lesssim t^4$,
    \[ 0 \le \E \bX_{i}^2 - q_{\beta,h_i} \lesssim \E \bX_{i}^4 \lesssim \left(\E \bX_{i}^2\right)^2 \mper \]
    Therefore,
    \[ 1 - O\left( \E \bX_{i}^2 \right) \le \frac{q_{\beta,h_i}}{\E \bX_{i}^2} \le 1 \mcom \]
    so
    \[ \frac{q_{\beta,h_i}}{\E \bX_{i}^2} \to 1 \mper \]
    However, $\frac{q}{\E \left( \beta\bg\sqrt{q} + h \right)^2} \le \frac{1}{\beta^2} < 1$, a contradiction.

    It remains to prove (iii). Define the shorthand $\bT = \tanh(\beta \bg \sqrt{q} + h)$ and $\bS = \sech(\beta \bg \sqrt{q} + h)$. 
    We can implicitly differentiate through $h$ to obtain
    \begin{align*}
        q' &= \E[2\bT\bS^2(\beta \bg (\sqrt{q})' + 1)] \\
        &= 2\E[\bT\bS^2] + \frac{\beta q'}{\sqrt{q}} \E[\bg \bT\bS^2] \\
        &= 2\E[\bT\bS^2] + \beta^2 q' \E[\bS^4 - 2\bS^2\bT^2]\mper\tag{Stein's lemma}
    \end{align*}
    Solving for $q'$ yields 
    \[
    q' = \frac{2\E[\bT\bS^2]}{1 - \beta^2\E[\bS^4 - 2\bS^2\bT^2]} \mper
    \]
    For $(\beta,h)$ satisfying \eqref{eq:at-line}, the denominator is strictly positive (thus justifying the well-definedness of $q'$). The numerator is strictly positive by \Cref{fact:sech-tanh}. \qedhere
\end{proof}

\watmonotonicity*
\begin{proof}
    This is immediate from \Cref{lem:q-properties}(iii), and the fact that \eqref{eq:at-line} may be rewritten as $\beta^2(1-q) < 1$.
\end{proof}

We will also need the following estimates for the behavior of $q$ and $q_1$ for small $h$.
\begin{lemma}\label{lem:q1-q-bound}
    Let $h \ge 0$ and $\beta < 1$. Also let $q(h)$ be the unique solution to $q = \E[\tanh^2(\beta \bg \sqrt{q} + h)]$ and define $q_1(h) = \E[\tanh(\beta \bg \sqrt{q(h)} + h)]$. Then 
    \begin{align*}
         \frac{h^2}{1-\beta^2} - \frac{2h^4}{(1-\beta^2)^3} &\le q(h) \le \frac{h^2}{1-\beta^2} \\
         h - \frac{h^2}{3(1-\beta^2)}& \le q_1(h) \le h
    \end{align*}
\end{lemma}
\begin{proof}
    The inequality for $q$ is \cite[Lemma 1.2]{DW23}. For the inequality on $q_1$, we first note that $\tanh x \le x$, so with $Y = \beta \sqrt{q} \bz + h$, we have $q_1 = \E \tanh Y \le \E Y = h$. Furthermore, since $x - \tanh x \le \frac{x^2}{3}$, we obtain $\E[Y - \tanh Y] \le \frac{1}{3} \E Y^2 = \frac{1}{3}(h^2 + \beta^2 q)$. Combining with the upper bound on $q$, we obtain 
    \[
    q_1 \ge h - \frac{h^2}{3(1-\beta^2)} \mcom\]  as desired.
\end{proof}

We may now establish some nice properties of $f_{\beta}$.

\begin{lemma}\label{lem:fixed-point}
    Fix $\lambda > 1$ and $\beta > 0$. 
    Consider the function $f_{\beta}$ defined as in \eqref{eq:f-beta-lambda}:
    \[ f_{\beta}\left( h \right) = \E_{\bg \sim \calN\left( 0,1 \right)} \tanh\left( \beta \bg \sqrt{q_{\beta,\beta\lambda h}} + \beta \lambda h \right) \mper \]
    Then,
    \begin{enumerate}[label=\normalfont{(\roman*)}]
        \item If $(\beta,\beta\lambda h)$ satisfies \eqref{eq:at-line}, the function $m \mapsto \frac{f_{\beta}(m)}{m}$ is strictly decreasing at $h$.
        \item If $(\beta,\beta\lambda h)$ satisfies \eqref{eq:at-line}, $f_{\beta}$ is increasing at $h$.
    \end{enumerate}
\end{lemma}

\begin{proof}
    Let $\bX = \bX(h) = \beta \lambda h + \beta \bg \sqrt{q_{\beta,\beta\lambda h}}$, and introduce the function $r = r(h) = q_{\beta,\beta\lambda h}$. By Stein's lemma, we have
    \begin{align*}
        f_{\beta}'(h) &= \E_{\bg \sim \calN\left( 0,1 \right)} \qty[\sech^2 \left( \bX \right) \cdot \left( \beta  \bg \cdot \dv{}{h} \sqrt{r} + \beta \lambda \right)] \\
        &= \E \qty[\sech^2(\bX)\qty(\beta \bg \frac{r'}{2\sqrt{r}} + \beta \lambda)] \\
        &= \beta\lambda \E[\sech^2(\bX)] - \beta^2 r'\E\qty[\sech^2(\bX)\tanh(\bX)] \mper
    \end{align*}
    Here, we have implicitly used the differentiability of $\sqrt{r(h)}$, which we justify now.

    This is not an issue when $\beta \ge 1$, because \eqref{eq:at-line} implies that $h > 0$, and thus $r > 0$. If $\beta < 1$, by \Cref{lem:q-properties}(iii), $r$ is differentiable and strictly positive for all $h > 0$. Let us now compute the limit as $h \to 0^+$.
    By \Cref{lem:q1-q-bound,lem:q-properties}, we have $r'(h) = \beta^2\lambda^2 \cdot (1+o_h(1))\frac{2h}{1-\beta^2}$ and $r(h) = \beta^2\lambda^2 \cdot \frac{h^2}{1-\beta^2}(1+o_h(1))$.
    Hence,  
    \begin{align*}
        \lim_{h \to 0^+} \frac{r'(h)}{\sqrt{r(h)}} = \sqrt{\frac{2\beta^2\lambda^2}{1 - \beta^2}}\mper
    \end{align*}
    It follows by symmetry that $f'_\beta(h)$ is well defined for $\beta < 1$ and all $h \in \R$.
        
    Let us start by establishing (ii). Abbreviate $\bS = \sech(\bX)$ and $\bT = \tanh(\bX)$. We have by \Cref{lem:q-properties}(ii) and \eqref{eq:at-line} that
    \begin{align*}
        f_{\beta}'(h) &= \beta\lambda \E \bS^2 - \beta^2 \cdot \beta\lambda \cdot \frac{2 \E[\bT\bS^2] }{1 - \beta^2\E \left[ \bS^4 - 2 \bS^2 \bT^2 \right]} \cdot \E[\bT\bS^2] \\
        &\ge \beta\lambda \left( \E \bS^4 - \frac{\E[\bS^2\bT]^2}{\E[\bS^2\bT^2]} \right) \ge 0 \mcom
    \end{align*}
    where the final inequality is the Cauchy--Schwarz inequality. 
    
    Now consider (i). Set $g(h) = \frac{f_{\beta}(h)}{h}$. We have
    \begin{align}
        h^2 \cdot g'(h) &= f'_{\beta}(h) h - f_{\beta}(h) \nonumber \\
        &= \E[\beta\lambda h \bS^2] - \beta^2 h r' \E\qty[\bS^2\bT] - \E[\bT]\mper \label{eq:convexity-equation-something-idk-man}
    \end{align}
    For fixed $h$, define the function $F(z) = \E \tanh(\beta\lambda z + \beta\bg\sqrt{q_{\beta,\beta\lambda h}})$, and denote $\bY = \beta\lambda z + \beta\bg\sqrt{q_{\beta,\beta\lambda h}}$. Note that $F$ is \emph{not} the same as $f_{\beta}$, since we have fixed $q = q_{\beta,\beta\lambda h}$ (in particular, $q \ne q_{\beta,\beta\lambda z}$). We may then compute for $z > 0$ that
    \begin{align*}
        F'(z) &= \beta\lambda \E \sech^2(\bY) > 0 \\
        F''(z) &= -2 \beta^2\lambda^2 \E [ \sech^2(\bY) \tanh(\bY) ] < 0 \mcom
    \end{align*}
    where the final inequality is \Cref{fact:sech-tanh}. Plugging this back into \eqref{eq:convexity-equation-something-idk-man},
    \[ h^2 \cdot g'(h) = -\beta^2 h \cdot r' \cdot \E[\bS^2 \bT] + h F'(h) - F(h) \mper \]
    \Cref{lem:q-properties}(iii) and \Cref{fact:sech-tanh} imply that the first term is negative. For the other two terms, note that if we set $G(z) = z F'(z) - F(z)$, then because $G'(z) = F''(z) < 0$ and $G(0) = 0$, we have $G(h) < 0$, completing the proof.
\end{proof}

Let us next establish the (in)stability of the $0$ fixed point.

\begin{lemma}
    \label{lem:zero-stability-fbeta}
    If $\beta < \frac{1}{\lambda}$, then
    \[ \lim_{m \downarrow 0} \frac{f_{\beta}(m)}{m} < 1 \mper \]
    In particular, $0$ is a stable fixed point. If $\lambda \ge \beta > \frac{1}{\lambda}$, then
    \[ \lim_{m \downarrow 0} \frac{f_{\beta}(m)}{m} > 1 \mper \]
\end{lemma}
\begin{proof}
    By the fundamental theorem of calculus,
    \[ \frac{f_{\beta}(m)}{m} = \beta\lambda \int_{0}^{1} \E \left[ \sech^2 \left( \beta\bg\sqrt{q_{\beta,\beta\lambda m}} + t\beta\lambda m \right) \right] \dif t \mper \]
    When $\beta \le 1$, this essentially concludes the proof: by the dominated convergence theorem and \Cref{lem:q-properties}, the above integral converges to $\beta\lambda$, which is less than $1$ if $\beta < \frac{1}{\lambda}$, and greater than $1$ if $\beta > \frac{1}{\lambda}$. For $\beta > 1$, dominated convergence yields that the limit is equal to
    \[ \beta\lambda \E\left[ \sech^2\left( \beta\bg\sqrt{\ol{q}_{\beta}} \right) \right] \mper \]
    We must show that this is greater than $1$. Indeed, by the Gaussian \Poincare inequality, we have
    \[ \Var\left[ \tanh\left( \beta\bg\sqrt{\ol{q}_{\beta}} \right) \right] \le \beta^2 \ol{q} \cdot \E\left[ \sech^4\left( \beta\bg\sqrt{\ol{q}_{\beta}} \right) \right] \mper \]
    The left hand side is equal to $\E\left[ \tanh^2\left( \beta\bg\sqrt{\ol{q}_{\beta}} \right) \right] = \ol{q}$. As a result,
    \[ \beta\lambda\E\left[ \sech^2\left( \beta\bg\sqrt{\ol{q}_{\beta}} \right) \right] > \beta^2 \E\left[ \sech^4\left( \beta\bg\sqrt{\ol{q}_{\beta}} \right) \right] \ge 1 \mcom \]
    completing the proof.
\end{proof}

\begin{lemma}
    \label{lem:fixed-point-lies-in-wat}
    Let $h \ne 0$ and $\lambda \ge \beta$ such that $f_{\beta}(h) = h$. Then, $(\beta,\beta\lambda h)$ satisfies \eqref{eq:at-line}.
\end{lemma}
\begin{proof}
    Fix some fixed point $h_* \ne 0$ and $q_* = q_{\beta,\beta\lambda h_*}$. Define the function $g : \R \to \R$ by
    \[ g(s) = \E \left[\tanh\left( \beta \bg \sqrt{q_*} + s \right)\right] \mper \]
    In particular, by definition, $g(\beta\lambda h_*) = h_*$. A straightforward computation establishes that
    \[ g''(s) = -2 \E \left[ \tanh(\beta\bg \sqrt{q_*} + s) \cdot \sech^2\left(\beta\bg\sqrt{q_*} + s\right) \right] \mper \]
    By \Cref{fact:sech-tanh}, $g''(s) < 0$ for all $s > 0$. As a result, $g(s) > sg'(s)$ for all $s > 0$. Plugging in $s = \beta\lambda h_*$, we have
    \[ \beta\lambda \cdot g'(\beta\lambda h_*) < 1 \mcom \]
    and thus,
    \[ \beta^2 \E\left[ \sech^2\left( \beta\bg\sqrt{q_*} + \beta\lambda h_* \right) \right] \le \beta\lambda \E\left[ \sech^2\left( \beta\bg\sqrt{q_*} + \beta\lambda h_* \right) \right] = \beta\lambda g'(\beta\lambda h_*) < 1 \mcom \]
    completing the proof.
\end{proof}

To conclude this section, let us put the pieces together and prove \Cref{lem:fixed-pt-structure-unconditional}.

\begin{proof}[Proof of \Cref{lem:fixed-pt-structure-unconditional}]
    \Cref{lem:zero-stability-fbeta} coupled with \Cref{lem:fixed-point} imply that when $\beta < \frac{1}{\lambda} < 1$, $0$ is the only fixed point, and furthermore that it is stable. When $\beta > \frac{1}{\lambda}$, \Cref{lem:zero-stability-fbeta} says that $\lim_{m \to 0} \frac{f_{\beta}(m)}{m} > 1$. Since $\lim_{m \to \infty} f_{\beta}(m) \le 1$, there must exist at least one positive. \Cref{lem:fixed-point-lies-in-wat} then implies that for any such fixed point $h$, $(\beta,\beta\lambda h)$ satisfies \eqref{eq:at-line}. Finally, \Cref{cor:wat-continuous} and \Cref{lem:fixed-point}(i) imply that there cannot exist two such fixed points; indeed, $m \mapsto \frac{f_{\beta}(m)}{m}$ is strictly decreasing in the interval between them, a contradiction.
\end{proof}

\section*{Acknowledgments}

We are very grateful to Brice Huang for numerous discussions about the SK model, and to Hang Du, Brice Huang, and Mark Sellke for showing us the trick used in \Cref{lem:mag-second-moment}.
We thank Ahmed El Alaoui for helpful discussions about the SK model under weak external field. 
We benefited from countless conversations with Prasad Raghavendra and Sidhanth Mohanty, which inspired this work.
A.R. would also like to thank Sam Hopkins and Kuikui Liu for feedback on an earlier draft of this work.
We acknowledge the use of ChatGPT for literature search and helping prove several key lemmas.

\bibliographystyle{alpha}
\bibliography{main}

@article{Lop26,
  title={Replica symmetry up to the de Almeida-Thouless line in the Sherrington-Kirkpatrick model},
  author={Lopatto, Patrick},
  journal={arXiv preprint arXiv:2604.11921},
  year={2026}
}

@article{RSSLMRF25,
  title={Taming Imperfect Process Verifiers: A Sampling Perspective on Backtracking},
  author={Rohatgi, Dhruv and Shetty, Abhishek and Saless, Donya and Li, Yuchen and Moitra, Ankur and Risteski, Andrej and Foster, Dylan J},
  journal={arXiv preprint arXiv:2510.03149},
  year={2025}
}

@article{JT17,
  title={{Some Properties of the Phase Diagram for Mixed $p$-Spin Glasses}},
  author={Jagannath, Aukosh and Tobasco, Ian},
  journal={Probability Theory and Related Fields},
  volume={167},
  number={3},
  pages={615--672},
  year={2017},
  publisher={Springer}
}

@article{Ton02,
  title={{About the Almeida-Thouless transition line in the Sherrington-Kirkpatrick mean-field spin glass model}},
  author={Toninelli, Fabio Lucio},
  journal={Europhysics letters},
  volume={60},
  number={5},
  pages={764},
  year={2002},
  publisher={IOP Publishing}
}

@article{Gue01,
  title={Sum Rules for the Free Energy in the Mean Field Spin Glass},
  author={Guerra, Francesco},
  journal={Mathematical Physics in Mathematics and Physics: Quantum and Operator Algebraic Aspects},
  volume={30},
  pages={161},
  year={2001},
  publisher={American Mathematical Soc.}
}

@article{BCV25,
  title={A Structural Theory of Quantum Metastability: Markov Properties and Area Laws},
  author={Bergamaschi, Thiago and Chen, Chi-Fang and Vazirani, Umesh},
  journal={arXiv preprint arXiv:2510.08538},
  year={2025}
}

@article{Hub59,
  title={Calculation of partition functions},
  author={Hubbard, John},
  journal={Physical Review Letters},
  volume={3},
  number={2},
  pages={77},
  year={1959},
  publisher={APS}
}

@inproceedings{Str57,
  title={On a method of calculating quantum distribution functions},
  author={Stratonovich, RL},
  booktitle={Soviet Physics Doklady},
  volume={2},
  pages={416},
  year={1957}
}

@article{MW23,
  title={Posterior sampling in high dimension via diffusion processes},
  author={Montanari, Andrea and Wu, Yuchen},
  journal={arXiv preprint arXiv:2304.11449},
  year={2023}
}

@article{LW22,
  title={A non-asymptotic framework for approximate message passing in spiked models},
  author={Li, Gen and Wei, Yuting},
  journal={arXiv preprint arXiv:2208.03313},
  year={2022}
}

@article{RV18,
  title={Finite sample analysis of approximate message passing algorithms},
  author={Rush, Cynthia and Venkataramanan, Ramji},
  journal={IEEE Transactions on Information Theory},
  volume={64},
  number={11},
  pages={7264--7286},
  year={2018},
  publisher={IEEE}
}

@article{MV21,
  title={Estimation of low-rank matrices via approximate message passing},
  author={Montanari, Andrea and Venkataramanan, Ramji},
  year={2021}
}

@article{Mio17,
  title={Fundamental limits of low-rank matrix estimation: the non-symmetric case},
  author={Miolane, L{\'e}o},
  journal={arXiv preprint arXiv:1702.00473},
  year={2017}
}

@inproceedings{LM17,
  title={Fundamental limits of symmetric low-rank matrix estimation},
  author={Lelarge, Marc and Miolane, L{\'e}o},
  booktitle={Conference on Learning Theory},
  pages={1297--1301},
  year={2017},
  organization={PMLR}
}

@article{HR04,
  title={Principal-component-analysis eigenvalue spectra from data with symmetry-breaking structure},
  author={Hoyle, David C and Rattray, Magnus},
  journal={Physical Review E},
  volume={69},
  number={2},
  pages={026124},
  year={2004},
  publisher={APS}
}

@article{BBAP05,
  title={Phase transition of the largest eigenvalue for nonnull complex sample covariance matrices},
  author={Baik, Jinho and {Ben Arous}, G{\'e}rard and P{\'e}ch{\'e}, Sandrine},
  year={2005}
}

@article{AKUZ19,
  title={{Approximate Survey Propagation for Statistical Inference}},
  author={Antenucci, Fabrizio and Krzakala, Florent and Urbani, Pierfrancesco and Zdeborov{\'a}, Lenka},
  journal={Journal of Statistical Mechanics: Theory and Experiment},
  volume={2019},
  number={2},
  pages={023401},
  year={2019},
  publisher={IOP Publishing}
}

@inproceedings{LMRW24,
  title={{Fast Mixing in Sparse Random Ising Models}},
  author={Liu, Kuikui and Mohanty, Sidhanth and Rajaraman, Amit and Wu, David X},
  booktitle={2024 IEEE 65th Annual Symposium on Foundations of Computer Science (FOCS)},
  pages={120--128},
  year={2024},
  organization={IEEE}
}

@article{CT21,
  title={{On convergence of the cavity and Bolthausen’s TAP iterations to the local magnetization}},
  author={Chen, Wei-Kuo and Tang, Si},
  journal={Communications in Mathematical Physics},
  volume={386},
  number={2},
  pages={1209--1242},
  year={2021},
  publisher={Springer}
}

@incollection{Han07,
  title={A Limit Theorem for Mean Magnetisation in the Sherrington-Kirkpatrick Model with an External Field},
  author={Hanen, Albert},
  booktitle={Spin Glasses: Statics and Dynamics: Summer School, Paris 2007},
  pages={177--201},
  year={2007},
  publisher={Springer}
}

@article{HMRW24,
  title={{Weak Poincaré Inequalities, Simulated Annealing, and Sampling from Spherical Spin Glasses}},
  author={Huang, Brice and Mohanty, Sidhanth and Rajaraman, Amit and Wu, David X},
  journal={arXiv preprint arXiv:2411.09075},
  year={2024}
}

@misc{LMRRW24,
      title={{Locally Stationary Distributions: A Framework for Analyzing Slow-Mixing Markov Chains}}, 
      author={Kuikui Liu and Sidhanth Mohanty and Prasad Raghavendra and Amit Rajaraman and David X. Wu},
      year={2024},
      eprint={2405.20849},
      archivePrefix={arXiv},
      primaryClass={cs.DS}
}

@article{EKZ22,
  title={{A spectral condition for spectral gap: fast mixing in high-temperature Ising models}},
  author={Eldan, Ronen and Koehler, Frederic and Zeitouni, Ofer},
  journal={Probability theory and related fields},
  volume={182},
  number={3},
  pages={1035--1051},
  year={2022},
  publisher={Springer}
}

@inproceedings{CE22,
	author = {Yuansi Chen and Ronen Eldan},
	booktitle = {2022 IEEE 63rd Annual Symposium on Foundations of Computer Science (FOCS)},
	title = {{Localization Schemes: A Framework for Proving Mixing Bounds for Markov Chains}},
	year = {2022},
	volume = {},
	issn = {},
	pages = {110-122},
	doi = {10.1109/FOCS54457.2022.00018},
	url = {https://doi.ieeecomputersociety.org/10.1109/FOCS54457.2022.00018},
	publisher = {IEEE Computer Society},
	address = {Los Alamitos, CA, USA},
	month = {Nov}
}

@article{DKMZ11,
  title = {Asymptotic analysis of the stochastic block model for modular networks and its algorithmic applications},
  author = {Decelle, Aurelien and Krzakala, Florent and Moore, Cristopher and Zdeborov\'a, Lenka},
  journal = {Phys. Rev. E},
  volume = {84},
  issue = {6},
  pages = {066106},
  numpages = {19},
  year = {2011},
  month = {Dec},
  publisher = {American Physical Society},
  doi = {10.1103/PhysRevE.84.066106},
  url = {https://link.aps.org/doi/10.1103/PhysRevE.84.066106}
}

@inproceedings{EAMS22,
  title={{Sampling from the Sherrington-Kirkpatrick Gibbs measure via algorithmic stochastic localization}},
  author={El Alaoui, Ahmed and Montanari, Andrea and Sellke, Mark},
  booktitle={2022 IEEE 63rd Annual Symposium on Foundations of Computer Science (FOCS)},
  pages={323--334},
  year={2022},
  organization={IEEE}
}

@article{HMP24,
  title={Sampling from spherical spin glasses in total variation via algorithmic stochastic localization},
  author={Huang, Brice and Montanari, Andrea and Pham, Huy Tuan},
  journal={arXiv preprint arXiv:2404.15651},
  year={2024}
}

@inproceedings{AKV24,
  title={{Trickle-Down in Localization Schemes and Applications}},
  author={Anari, Nima and Koehler, Frederic and Vuong, Thuy-Duong},
  booktitle={Proceedings of the 56th Annual ACM Symposium on Theory of Computing},
  pages={1094--1105},
  year={2024}
}

@article{LFW23,
  title={{Approximate message passing from random initialization with applications to $\mathbb{Z}_2$ synchronization}},
  author={Li, Gen and Fan, Wei and Wei, Yuting},
  journal={Proceedings of the National Academy of Sciences},
  volume={120},
  number={31},
  pages={e2302930120},
  year={2023},
  publisher={National Acad Sciences}
}

@book{Tal10,
  title={Mean field models for spin glasses: Volume I: Basic examples},
  author={Talagrand, Michel},
  volume={54},
  year={2010},
  publisher={Springer Science \& Business Media}
}

@book{Tal11,
    title={Mean Field Models for Spin Glasses: Volume II: Advanced replica-symmetry and low temperature},
author={Talagrand, Michel},
  volume={55},
  year={2011},
  publisher={Springer Science \& Business Media}}

@article{DW23,
  title={Mean field spin glass models under weak external field},
  author={Dey, Partha S and Wu, Qiang},
  journal={Communications in Mathematical Physics},
  volume={402},
  number={2},
  pages={1205--1258},
  year={2023},
  publisher={Springer}
}

@article{DAM16,
author = {Deshpande, Yash and Abbe, Emmanuel and Montanari, Andrea},
year = {2016},
month = {07},
pages = {185-189},
title = {Asymptotic mutual information for the binary stochastic block model},
doi = {10.1109/ISIT.2016.7541286}
}

@article{PWBM18,
  title={Optimality and sub-optimality of PCA I: Spiked random matrix models},
  author={Perry, Amelia and Wein, Alexander S and Bandeira, Afonso S and Moitra, Ankur},
  journal={The Annals of Statistics},
  volume={46},
  number={5},
  pages={2416--2451},
  year={2018},
  publisher={JSTOR}
}

@article{EKJ20,
  title={Fundamental limits of detection in the spiked Wigner model},
  author={El Alaoui, Ahmed and Krzakala, Florent and Jordan, Michael},
  year={2020}
}

@article{DGPZ25,
  title={Sequential Dynamics in Ising Spin Glasses},
  author={Dandi, Yatin and Gamarnik, David and Pernice, Francisco and Zdeborov{\'a}, Lenka},
  journal={arXiv preprint arXiv:2506.09877},
  year={2025}
}

@article{Sel24,
  title={The threshold energy of low temperature Langevin dynamics for pure spherical spin glasses},
  author={Sellke, Mark},
  journal={Communications on Pure and Applied Mathematics},
  volume={77},
  number={11},
  pages={4065--4099},
  year={2024},
  publisher={Wiley Online Library}
}

@article{CK93,
  title={Analytical solution of the off-equilibrium dynamics of a long-range spin-glass model},
  author={Cugliandolo, Leticia F and Kurchan, Jorge},
  journal={Physical Review Letters},
  volume={71},
  number={1},
  pages={173},
  year={1993},
  publisher={APS}
}

@article{BDG06,
  title={{Cugliandolo-Kurchan equations for dynamics of spin-glasses}},
  author={{Ben Arous}, G{\'e}rard and Dembo, Amir and Guionnet, Alice},
  journal={Probability theory and related fields},
  volume={136},
  number={4},
  pages={619--660},
  year={2006},
  publisher={Springer}
}

@article{LSS22,
  title={High-dimensional asymptotics of Langevin dynamics in spiked matrix models},
  author={Liang, Tengyuan and Sen, Subhabrata and Sur, Pragya},
  journal={arXiv preprint arXiv:2204.04476},
  year={2022}
}

@article{DMM09,
  title={Message-passing algorithms for compressed sensing},
  author={Donoho, David L and Maleki, Arian and Montanari, Andrea},
  journal={Proceedings of the National Academy of Sciences},
  volume={106},
  number={45},
  pages={18914--18919},
  year={2009},
  publisher={National Academy of Sciences}
}

@article{BM11,
  title={The dynamics of message passing on dense graphs, with applications to compressed sensing},
  author={Bayati, Mohsen and Montanari, Andrea},
  journal={IEEE Transactions on Information Theory},
  volume={57},
  number={2},
  pages={764--785},
  year={2011},
  publisher={IEEE}
}

@book{BH16,
  title={Metastability: a potential-theoretic approach},
  author={Bovier, Anton and Den Hollander, Frank},
  volume={351},
  year={2016},
  publisher={Springer}
}

@article{BJ24,
  title={Shattering versus metastability in spin glasses},
  author={{Ben Arous}, G{\'e}rard and Jagannath, Aukosh},
  journal={Communications on Pure and Applied Mathematics},
  volume={77},
  number={1},
  pages={139--176},
  year={2024},
  publisher={Wiley Online Library}
}

@article{BGP24,
  title={{Langevin dynamics for high-dimensional optimization: the case of multi-spiked tensor PCA}},
  author={{Ben Arous}, G{\'e}rard and Gerbelot, C{\'e}dric and Piccolo, Vanessa},
  journal={arXiv preprint arXiv:2408.06401},
  year={2024}
}

@inproceedings{DM14,
  title={{Information-theoretically optimal sparse PCA}},
  author={Deshpande, Yash and Montanari, Andrea},
  booktitle={2014 IEEE International Symposium on Information Theory},
  pages={2197--2201},
  year={2014},
  organization={IEEE}
}

@article{MR15,
  title={Non-negative principal component analysis: Message passing algorithms and sharp asymptotics},
  author={Montanari, Andrea and Richard, Emile},
  journal={IEEE Transactions on Information Theory},
  volume={62},
  number={3},
  pages={1458--1484},
  year={2015},
  publisher={IEEE}
}

@article{MR14,
  title={{A statistical model for tensor PCA}},
  author={Montanari, Andrea and Richard, Emile},
  journal={Advances in neural information processing systems},
  volume={27},
  year={2014}
}

@article{BGJ20,
  title={{Algorithmic thresholds for tensor PCA}},
  author={{Ben Arous}, Gerard and Gheissari, Reza and Jagannath, Aukosh},
  journal={The Annals of Probability},
  volume={48},
  number={4},
  pages={2052--2087},
  year={2020},
  publisher={JSTOR}
}

@article{WEM19,
  title={{The Kikuchi hierarchy and tensor PCA}},
  author={Wein, Alexander and {El Alaoui}, Ahmed and Moore, Cristopher},
  journal={Journal of the ACM},
  year={2019},
  publisher={ACM New York, NY}
}

@inproceedings{LST21,
  title={Structured logconcave sampling with a restricted {G}aussian oracle},
  author={Lee, Yin Tat and Shen, Ruoqi and Tian, Kevin},
  booktitle={Conference on Learning Theory},
  pages={2993--3050},
  year={2021},
  organization={PMLR}
}

@article{STL20,
  title={{Composite Logconcave Sampling with a Restricted Gaussian Oracle}},
  author={Shen, Ruoqi and Tian, Kevin and Lee, Yin Tat},
  journal={arXiv preprint arXiv:2006.05976},
  year={2020}
}

@article{AT78,
  title={{Stability of the Sherrington-Kirkpatrick solution of a spin glass model}},
  author={de Almeida, Jairo RL and Thouless, David J},
  journal={Journal of Physics A: Mathematical and General},
  volume={11},
  number={5},
  pages={983},
  year={1978},
  publisher={IOP Publishing}
}

@book{MPV87,
  title={Spin glass theory and beyond: An Introduction to the Replica Method and Its Applications},
  author={M{\'e}zard, Marc and Parisi, Giorgio and Virasoro, Miguel Angel},
  volume={9},
  year={1987},
  publisher={World Scientific Publishing Company}
}

@article{KX25,
  title={Smooth Trade-off for Tensor PCA via Sharp Bounds for Kikuchi Matrices},
  author={Kothari, Pravesh K and Xu, Jeff},
  journal={arXiv preprint arXiv:2510.03061},
  year={2025}
}

\end{document}